\documentclass[journal]{IEEEtran}
\usepackage{citesort}
\usepackage[dvips]{graphicx}
\usepackage[cmex10]{amsmath}
\usepackage{algorithm}
\usepackage{algpseudocode}
\usepackage{array}
\usepackage[tight,footnotesize]{subfigure}
\usepackage{setspace}
\usepackage{enumerate} 
\usepackage{url}

\usepackage{amsfonts, amssymb, amsthm, mathrsfs}
\usepackage{setspace}
\usepackage{Tabbing}
\usepackage{fancyhdr}
\usepackage{lastpage}
\usepackage{extramarks}
\usepackage{chngpage}
\usepackage{soul,color}
\usepackage{indentfirst}
\usepackage{float,wrapfig}
\usepackage{times}
\usepackage{multirow}
\usepackage{bm}
\usepackage{bbm}
\usepackage{mathtools}

\usepackage{lipsum}

\newtheorem{theo}{Theorem}

\newtheorem{prop}{Proposition}

\newtheorem{lemma}{Lemma}

\newcolumntype{I}{!{\vrule width 1.5pt}}
\newlength\savedwidth

\newcommand{\sinr}{\textrm{SINR}}

\newcommand{\defeq}{\vcentcolon=}

\begin{document}

\title{Distributed Resource Allocation in Device-to-Device Enhanced Cellular Networks}

\author{Qiaoyang Ye, 
        Mazin Al-Shalash,
        Constantine Caramanis
        and Jeffrey G. Andrews
\thanks{Q. Ye, C. Caramanis and  J. G. Andrews are with the University of Texas at Austin, USA, M. Shalash is with Huawei Technologies. Email: 
qye@utexas.edu,  mshalash@huawei.com, constantine@utexas.edu, jandrews@ece.utexas.edu. 
Manuscript last revised: \today.}}

\maketitle

\begin{abstract}
Cellular network performance can significantly benefit from direct device-to-device (D2D) communication, but interference from  cochannel D2D communication limits the performance gain. In hybrid networks consisting of D2D and cellular links, finding the optimal interference management is challenging. In particular, we show that the problem of maximizing network throughput while guaranteeing predefined service levels to cellular users is non-convex and hence intractable. Instead, we adopt a distributed approach that is computationally extremely efficient, and requires minimal coordination, communication and cooperation among the nodes. The key algorithmic idea is a signaling mechanism that can be seen as a fictional pricing mechanism, that the base stations optimize and transmit to the D2D users, who then play a {\em best response} (i.e., selfishly) to this signal. Numerical results show that our algorithms converge quickly, have low overhead, and achieve a significant throughput gain, while maintaining the quality of cellular links at a predefined service level.
\end{abstract}

\IEEEpeerreviewmaketitle

\section{Introduction}
To meet the growing demand for local wireless services, direct communication between user equipments (UEs) -- called \emph{device-to-device (D2D) communication} -- is envisioned as an  important technology component for LTE-Advanced. Taking advantage of physical proximity, D2D enables a more flexible infrastructure than conventional cellular networks with potential benefits such as efficient resource utilization, large throughput and reduced end-to-end latency, exploiting which D2D can create new business opportunities by supporting various applications, such as  content sharing, multiplayer gaming, social networking services and  mobile advertising \cite{DopRin09,CorLar10,WuTav10,FodDah12,LinAnd13D2Doverview,Qualcomm14}. In this paper, we focus on the use cases where the D2D traffic is generated by UEs themselves (e.g., sharing a just-taken video), and thus only consider one hop D2D transmission, which is different from the case where D2D works as a relay \cite{NisIto14}.
 In D2D-enabled cellular networks,  resources can be either allocated orthogonally to D2D and cellular links, or shared between them. While the interference management is certainly simplified in a network with orthogonal allocation, the resource utilization is less efficient. To improve the resource utilization, we consider a hybrid network with shared  allocation -- called a \emph{shared network} in this paper.

The uplink spectrum is often under-utilized compared to downlink \cite{MarBas08}, and thus letting D2D links use uplink resources may improve the resource utilization. Moreover, when D2D links share downlink resources, base stations (BSs) become fairly strong interferers for D2D receivers, and D2D transmitters may cause high interference to nearby co-channel cellular UEs, which may significantly degrade the network performance \cite{YeAndHop13}. In contrast, when D2D links use uplink resources,  the interference from D2D to cellular transmissions  can be better handled, since the BSs that are more powerful than UEs suffer from D2D interference. Therefore, sharing uplink spectrum is preferable overall~\cite{LinAnd13D2Doverview}.

The objective of this paper is to improve the network throughput by allowing D2D communication to share the cellular uplink resources. At the same time, we restrict the access of D2D links to the uplink spectrum in order to manage the interference.
Generally, centralized interference management requires the central controller (e.g., the BS) to acquire the channel state information (CSI) between each transmitter and receiver. This requires high overhead, particularly in the scenario where channels vary rapidly (e.g., in a high-mobility environment).  Therefore, a distributed algorithm requiring only local information  is preferable. We address the following design question in this paper: how to intelligently manage spectrum for D2D with only local information and BS assistance (e.g., setting a high cost for D2D links causing strong interference), so as to manage the interference and improve the network throughput? 

\subsection{Related Work}
There has been increasing interest recently in the investigation of interference management in shared networks.  Power control is one viable approach, e.g., consideration of a greatly simplified model with one cellular UE and one D2D link \cite{YuDop11}, a simple power reduction method based on the derived signal-to-interference-plus-noise ratio (SINR) \cite{YuTir09},  and study of several power control schemes including fixed power and fixed SNR target \cite{GuBae11} are some existing works. Another popular approach, related to the direction we propose, is to  intelligently manage spectrum for D2D links based on channel conditions and nearby interfering UEs, e.g., maximal mutual interference minimization~\cite{JanKoi09}, network throughput maximization~\cite{ZulHua10}, setting exclusive D2D transmission zone to achieve interference avoidance \cite{XuWan10,MinLee11}, and interference randomization through frequency and/or time hopping \cite{CheCha10,YeAndHop13}.  
However,  the key new aspect of D2D-enabled cellular networks -- assistance of BSs -- has received much less treatment in the literature. {Possibly one reason for this is that the computational problem itself is quite difficult, and the communication and coordination alone required for a good centralized solution might be prohibitive.}

As we discuss below, the key approach adopted in this paper is a two-stage distributed algorithm that has a game theoretic interpretation: BSs send out a signal representing a fictional price that can be considered as the assistance from BSs, and then D2D users optimize a local objective function adapting to that price, and  to what the other users are doing.  Since individual D2D links optimize their local functions ``selfishly'', this approach has a game theoretic interpretation, and this allows us to use algorithms and analysis concepts from game theory, even though there is no actual market, and the users agree to ``play'' this game using the BS's signal, without actually exchanging currency.

Game theory has been used in various disciplines to model competition for limited resources in more general networks. For example, work has been done considering spectrum sharing based on local bargaining~\cite{CaoZhe05}, repeated game \cite{EtkPar07}, auction mechanisms \cite{XuSon12,XuSon13} and two-stage game  \cite{RazLuo11,WanSon13}. Paper~\cite{SonNiy14} demonstrates several different game models for D2D resource allocation, where an interesting example is to use the reverse iterative combinatorial auction \cite{XuSon12globecom}. In this paper, the game theoretic approach is used as an {\em algorithmic technique} to obtain efficient distributed spectrum management. Similar to the recent work \cite{WanSon13}, we model a Stackelberg game to control the interference from D2D to cellular network. The key difference from \cite{WanSon13} is in the upper-stage problem, where we takes into account the D2D rate. Moreover, we investigate the convergence of the algorithms for both lower-stage and upper-stage problems. Works in cognitive radio such as~\cite{RazLuo11} are related to the second part of our work (i.e., the investigation of optimal prices to charge D2D links accessing the shared resources), which proposes that secondary users adapt their powers for alleviating interference to primary users. The key techniques used in this paper to study the optimal prices are similar to \cite{RazLuo11}, but we in addition investigate the convergence of the spectrum management scheme for the D2D network (i.e., the lower-stage problem), as well as the convergence of the proposed heuristic algorithm.

\subsection{Contributions And Organizations}
{We present a distributed, efficient and low-overhead spectrum management method for D2D links to improve the  throughput while keeping the performance of cellular users at a guaranteed level. Specifically, the main contributions are:}

\textbf{A Two-stage Distributed Algorithm.}
We propose an iterative two-stage algorithm in Section~\ref{sec:formulation}. In the first stage, the BSs send a pricing signal that adapts to the gap between the aggregate interference from D2D links and a predefined interference tolerance level, where the price increases if D2D interference is higher than the tolerance level and decreases otherwise. In the second stage, each D2D link independently maximizes its utility consisting of a reward equal to its expected rate and a penalty proportional to the interference caused by this link to the BS, as measured by the pricing signal.  Note that this two-stage model is  a Stackelberg game \cite{Osb94}, and the algorithm can be seen as a pricing mechanism. This algorithm requires no cooperation among D2D links, yet succeeds in discouraging strongly interfering or low-SINR D2D links to access more resource blocks (RBs).

\textbf{Utility-based D2D resource allocation adaptation.}
In Section \ref{sec:D2Dgame}, we consider the lower-stage problem, where we maximize the D2D rate in terms of expected SINR for tractability, which provides a performance upper bound and can serve as a benchmark. Each D2D link selfishly maximizes its utility given other D2D links' decisions and the price broadcast from BSs, which essentially forms a non-cooperative game. To reduce the computational complexity and overhead, we further consider the problem of maximizing a lower bound of the utility function for each D2D link. We then propose an efficient iterative algorithm similar to a waterfilling algorithm, which only requires local information. Our simulation results show that the result obtained by the proposed iterative algorithm is very close to the solution to the upper bound problem. This further lightens the computational burden on each user.

\textbf{Cellular link performance protection.}
Given the solution of the lower-stage problem, we study the optimal price the BSs report in Section \ref{sec:uppergame}, to maximize the network utility while protecting cellular links. We show that this problem can be transformed into a linear complementarity problem (LCP). This allows us to take advantage of, and adapt for our problem, general algorithms for LCP. We further propose a simpler  heuristic algorithm based on the bisection method, and observe that it has low overhead and converges very quickly with almost no loss. We also propose a simple greedy algorithm  that leads to efficient computation at the cost of overall throughput, where the throughput loss decreases as the interference tolerance level increases, e.g., the throughput loss compared to the algorithm for LCP are about $25\%$ and $5\%$ when the interference tolerance level is $5$dB lower and above the cellular signal in our setup,~respectively.

Numerical results  in Section \ref{sec:simulation} show that the cellular links can be well protected with the average D2D throughput reduction of only  $12\%$ in our setup, compared to the scenario where all D2D links are active. On the other hand, compared to conventional cellular networks without D2D communication, the proposed algorithms provide significant throughput gain (about $5$x with $10$ D2D links per cell and average D2D link length $80$m in our simulation setup). Note that the throughput gain highly depends on the amount of D2D traffic and average D2D link length. We take $10$ D2D links per cell and link length $80$m as an illustration example in this paper.


\section{System Model}\label{sec:model}
We consider a uplink shared network, where cellular UEs in the same cell get different sub-bands (i.e., orthogonal chunks of RBs). Any general scheduling scheme for cellular UEs can be used. A potential D2D link can either transmit directly by D2D communication, or transmit to a BS. We call the choice  \emph{mode selection}. 
By intelligently conducting mode selection, we can adjust the aggregate interference in the network and thus optimize the  achievable network performance.  However, the mode selection variables in the $\sinr$ expression result in non-convexity of objective functions that are in terms of rate, i.e., $\log(1+\sinr)$.  Moreover, the mode selection variables are binary, making the problem combinatorial. Note that different mode selection schemes lead to different optimal spectrum management, due to the differences in the resource allocation of cellular users and thus the differences in the interference tolerance level. On the other hand, spectrum management affects the achievable rate and thus affects the  mode selection of D2D links.  As discussed above, it is difficult to find the optimal mode selection, let alone the joint optimal mode selection and spectrum management of D2D links, where mode selection and spectrum management are coupled with each other. Despite the intractability of the optimization problem, there are various practical (but not necessarily optimal) mode selection approaches (e.g., distance-based mode selection \cite{LinAnd13uplinkD2D}).
 
We propose the following mode selection as one viable scheme. The potential D2D transmitters are treated the same as cellular UEs when scheduling, except that we can add a weight to the scheduling metric. For example, with proportional fair scheduling~\cite{Tse05}, the user with the largest $qR_i/\bar{R}_i$ would be scheduled, where $q$ is the weight, $R_i$ and $\bar{R}_i$ are the instantaneous rate and average rate of link $i$, respectively. Without loss of generality, we assume cellular users have $q=1$.  A typical value for $q$ of potential D2D links might be $1/2$, since a potential D2D link in cellular mode would occupy both uplink and downlink resources. We can also let the weight $q$ impose the cost on the backhual usage of core networks. Considering that D2D mode has more efficient resource utilization, the potential D2D links are biased  against cellular mode using $q<1$. 
Note that other mode selection schemes can also be applied to our following framework easily.

We assume potential D2D links that are not scheduled by BSs would be in D2D mode. In other words, we propose to let each BS complete the mode selection of potential D2D links in its coverage area. Given the mode selection, we aim to find the optimal spectrum management of D2D links to maximize the network utility.  We leave the joint optimization of mode selection and spectrum management to future work.

Assuming that the cellular resource allocation is done by the BSs, our focus is on the spectrum management of D2D links. In this paper, we consider resource allocation at each RB to simplify the notation and explanation, but any general units of RBs can be considered similarly. 
The sets of cellular UEs accessing the $k$th RB and of D2D links are denoted by $\mathcal{C}_k$ and $\mathcal{D}$, respectively, where the set of cellular UEs includes the potential D2D links in cellular mode. We define  $\mathcal{I}_k$ as the set of  D2D links accessing RB $k$ (i.e., the  set of interfering D2D links). Then the SINRs of D2D link $i$ at RB $k$ and a cellular UE belonging to $\mathcal{C}_k$ are, respectively,
\begin{equation}
\sinr_{{ik}}^{(D)} = \frac{\mathbbm{1}_{\{i\in \mathcal{I}_k\}} P_{D_i}h_{ii}^{(k)}}{\sum_{j\in\mathcal{I}_k, j\neq i}P_{D_j}h_{ji}^{(k)} +\sum_{j\in\mathcal{C}_k} P_{C_j}h_{ji}^{(k)}+W_{ik}},
\end{equation}
\begin{equation}
\sinr_{ik}^{(C)} = \frac{P_{C_i}g_{ii}^{(k)}}{\sum_{j\in\mathcal{I}_k}P_{D_j}g_{ji}^{(k)} +\sum_{j\in\mathcal{C}_k,j\neq i} P_{C_j}g_{ji}^{(k)}+W_{ik}},
\end{equation}
where $\mathbbm{1}_{\{a\in \mathcal{A}\}}$ is an indicator function with $\mathbbm{1}_{\{a\in \mathcal{A}\}}=1$ if $a\in\mathcal{A}$ and $\mathbbm{1}_{\{a\in \mathcal{A}\}}=0$ otherwise,  $P_{D_i}$ and $P_{C_i}$ are the transmit powers of D2D and  cellular links, respectively, $h_{ji}^{(k)}$ is the channel gain from UE $j$ to D2D receiver $i$ at RB $k$, $g_{ji}^{(k)}$ is the channel gain from UE~$j$ to the BS serving cellular UE $i$, and $W_{ik}$ is the noise power of link $i$ at RB $k$. We use Shannon capacity to calculate rate, i.e., $R_{ik}=B\log_2(1+\sinr_{ik})$, where $B$ is the frequency bandwidth of a RB. 

\section{Problem Formulation}\label{sec:formulation}
In this section, we first formulate a single-stage optimization problem to maximize the D2D throughput with a performance protection for cellular links. The computational intractability of the single-stage optimization then motivates us to consider a distributed setting, where each D2D link tries to maximize its own utility based only on local information.

\subsection{Single-stage Problem Formulation}
Without loss of generality, we let $x_{ik}$ be the probability that D2D link $i$ accessing RB $k$. The investigation of  optimal access probabilities upper bounds the channel assignment problem where D2D links either access a RB or not (i.e., the access probability is either 1 or 0). We consider the following utility maximization problem, subject to a  D2D interference constraint to guarantee the cellular performance: 
\begin{equation}\label{eq:1stage-opt}
\begin{aligned}
\max_{\mathbf{x}}  \ & \sum_{i\in\mathcal{D}} w_i \sum_{k=1}^K R_{ik}^{(D)}(\mathbf{x})\\
\text{s.t. } & \sum_{i\in\mathcal{D}} x_{ik} P_{D_i} g_{ii}^{(k)} \leq Q_k, \forall k,\\
& x_{ik} \in [0, 1],
\end{aligned}
\end{equation}
where $w_i$ is the weight for the $i$th D2D link, and $K$ is the number of total available RBs for D2D links. Denoting the power set of $\mathcal{D}$ by $2^{\mathcal{D}}$, the rate of D2D link $i$ at RB $k$ is
\begin{equation}\label{eq:RD}
\begin{aligned}
R_{ik}^{(D)} = \sum_{\mathcal{I}_k \in 2^{\mathcal{D}}} \prod_{j\in \mathcal{I}_k} x_{jk}\prod_{n\in\mathcal{D}\setminus\mathcal{I}_k}  (1-x_{nk}) \log_2\left( 1 +\sinr_{ik}^{(D)}\right).
\end{aligned}
\end{equation}
The first constraint in~(\ref{eq:1stage-opt}) is for  protection of cellular transmissions, where $Q_k$ -- called the \emph{interference tolerance level} -- depends on the channel condition of cellular transmission on RB $k$ (e.g., $Q_k$ could be the signal strength of the cellular link using RB $k$ multiplied by a predefined threshold). 
Note that $Q_k$ can be optimized to maximize a utility function incorporating the cellular rate. In this paper, we consider $Q_k$ as a predefined parameter and leave the joint optimization of $Q_k$ and D2D resource allocation to future work. 
We  observe that (\ref{eq:1stage-opt}) is not a convex optimization problem.
The computational complexity of a brute-force approach to solve (\ref{eq:1stage-opt}) is $O(N_x^{N_D} N_D^2)$, where $N_x$ is the number of possible values of $x$ to be searched and $N_D$ is the number of D2D links. Thus, the  computation is essentially impossible for even a modest-sized  network.

Instead of the centralized approach, we adopt a different strategy that results in an {\em efficient, distributed algorithm with low coordination, cooperation and communication overhead}. For tractability, we introduce variables -- called \emph{prices} for accessing RBs -- to decouple the interference constraint in  (\ref{eq:1stage-opt}) and develop a distributed tractable framework. Particularly, BSs adjust prices to control the total D2D interference, and each D2D link individually maximizes its utility in terms of the expected rate and prices charged by BSs. This leads to a two-stage optimization problem, which consists of  a problem to find optimal prices and several small-size convex optimization problems for D2D links. Though solutions to the two-stage problem may not provide the optimal solution to the original single-stage problem (\ref{eq:1stage-opt}), this relaxation allows us to efficiently allocate resources in a distributed fashion, and the numerical results in Section~\ref{sec:simulation} demonstrate a large rate gain without serious degradation in cellular performance using the proposed algorithm for the two-stage problem.

\subsection{Two-stage Problem Formulation}
We propose a pricing mechanism, where a BS charges the D2D link $i$ in its coverage area the amount $\mu_{ik}$ per unit of the interference caused by this D2D link to the BS at RB $k$, i.e., the cost for a D2D link to access  RB $k$ is $\mu_{ik}x_{ik}P_{D_i}g_{ii}^{(k)}$. 
Assuming that each cell runs this mechanism independently, the cost of a D2D link only depends on the interference caused by this D2D link to its associated BS. 

We assume the interference from other cells is invariant when we consider the resource allocation in a typical cell. Therefore, we can incorporate the interference from neighboring cells into noise and the multi-cell scenario is simplified to a single-cell scenario. Under this assumption, the interference constraint is for the interference caused by D2D links in this cell. Note that in this case, the updated noise (incorporating inter-cell interference) is different from user to user, where generally cell-edge users suffer larger noise. Though we focus on the asynchronous scheduling scenario, the proposed framework can be easily generalized to a synchronous multi-cell scenario if the price at each RB is unified among different cells, where the BS in the proposed model becomes a network controller, and the interference becomes the aggregate interference from D2D links to all BSs in the network.

The net utility of D2D link $i$ is $U_i=w_i \sum_{k=1}^K R_{ik}^{(D)}(\mathbf{x})- \sum_{k=1}^K \mu_{ik}x_{ik}P_{D_i}g_{ii}^{(k)}$, where the first and second term can be considered as the reward and penalty functions, respectively.
The problem involves a non-cooperative network, where each D2D link aims to maximize its utility selfishly. We denote the access probabilities of D2D link $i$ by  $\mathbf{x}_i\defeq [x_{i1}, x_{i2}, \cdots, x_{iK}]^T$. The access probabilities of all other D2D links are denoted by $\mathbf{x}_{-i}\defeq[\mathbf{x}_1^T, \cdots,\mathbf{x}_{i-1}^T, \mathbf{x}_{i+1}^T, \cdots, \mathbf{x}_{N_D}^T]^T$, where $N_D$ is the number of D2D links. Similarly, we define the \textit{price vector} of D2D link $i$ as~$\boldsymbol{\mu}_i\defeq[\mu_{i1}, \mu_{i2}, \cdots, \mu_{iK} ]^T$. Given $\boldsymbol{\mu}_{i}$ and $\mathbf{x}_{-i}$, the problem for the D2D link $i$ is
\begin{equation}\label{eq:opt-D2D}
\begin{aligned}
\max_{\mathbf{x}_{i}} \ & U_i(\mathbf{x}_{i}; \mathbf{x}_{-i}, \boldsymbol{\mu}_{i})\\
\text{s.t. } & x_{ik} \in [0, 1],\ \forall k.
\end{aligned}
\end{equation}
On the other hand, the network aims to find optimal prices:
\begin{equation}\label{eq:opt-BS}
\begin{aligned}
\max_{\boldsymbol{\mu}\geq 0} \ &U_c\left(\boldsymbol{\mu}, \mathbf{x}^*\left(\boldsymbol{\mu}\right)\right)\\
\text{s.t. } & \sum_{i\in\mathcal{D}}x_{ik}^*(\boldsymbol{\mu})P_{D_i}g_{ii}^{(k)} \leq Q_k,  \ \forall k,
\end{aligned}
\end{equation}
where $U_c\left(\boldsymbol{\mu}, \mathbf{x}^*\left(\boldsymbol{\mu}\right)\right) =\sum_{i\in\mathcal{D}}  \sum_{k=1}^K \mu_{ik}x^*_{ik}\left(\boldsymbol{\mu}\right)P_{D_i}g_{ii}^{(k)}$, and $\mathbf{x}^*\left(\boldsymbol{\mu}\right)$ is the solution of (\ref{eq:opt-D2D}) for a given $\boldsymbol{\mu}$. Taking a game theoretic perspective, the above problem is a  decentralized Stackelberg game (a two-stage game), where the leader moves first and then the followers move accordingly. In this paper, the BS is the leader and the D2D links are the followers.  


To solve the two-stage problem, we use a backward induction technique. We start with the problem of the D2D links -- called a \emph{lower problem} -- and get the D2D access probability $\mathbf{x}^*\left(\boldsymbol{\mu}\right)$. By plugging $\mathbf{x}^*\left(\boldsymbol{\mu}\right)$ into (\ref{eq:opt-BS}), we then investigate the network utility maximization -- called an \emph{upper problem}.

\section{Lower Problem: A Non-cooperative D2D Network}\label{sec:D2Dgame}
Given $\boldsymbol{\mu}$,  D2D links try to maximize their utility selfishly. This defines a non-cooperative game $G_D=[\mathcal{D}, \{\mathbf{x}_i\}, \{U_i\}]$. 
%
For tractability, we use Jensen's inequality and consider the following objective function that upper bounds (\ref{eq:RD}):
\begin{equation}\label{eq:D2Dopt}
\begin{aligned}
\max_{\mathbf{x}_i} \ & w_i \sum_{k=1}^K \tilde{R}_{ik}^{(D)} (\mathbf{x})- \sum_{k=1}^K \mu_{ik}x_{ik}P_{D_i}g_{ii}^{(k)}\\
\text{s.t. } & x_{ik}\in[0,1], \ \forall k,
\end{aligned}
\end{equation}
where 
\begin{equation*}
\tilde{R}_{ik}^{(D)} = \log\left( 1 + \sum_{\mathcal{I}_k \in 2^{\mathcal{D}}} \prod_{j\in \mathcal{I}_k}  x_{jk}\prod_{n\in\mathcal{D}\setminus\mathcal{I}_k} (1-x_{nk})\sinr_{ik}^{(D)}\right).
\end{equation*}
The upper bound is tight if most $x_{ik}$ are binary. We compare the gap between the solution maximizing (\ref{eq:RD}) and (\ref{eq:D2Dopt}) numerically in Section \ref{sec:simulation} and leave the analysis to future work.

We adopt an identical price for D2D links accessing the same RB, i.e., $\mu_{ik}=\mu_{jk}$. The rationale for doing this is that the BS only cares about the aggregate interference, rather than the differences between the interference values from different D2D links. The structure of (\ref{eq:D2Dopt}) suggests decoupling the lower problem into $K$ subproblems, where we consider each RB independently. In the rest of this paper, we consider a typical RB, and ignore the RB index $k$ for notation simplicity.

\subsection{Distributed Algorithms Design}\label{sec:D2Dgame1}
Optimization problems produce solutions with certain optimality guarantees. In our setting, however, the D2D links behave in a non-cooperative fashion. Thus, understanding the behavior and performance of our algorithm requires consideration of a different solution concept. This notion has been well-studied in game theory, and it is known that the analog of stationary points in an optimization solution are the so-called {\em Nash Equilibrium} (NE) points. In our context, these are the fixed points from which no D2D link would want to {\em unilaterally deviate} \cite{Osb94}. In the rest of this paper, the NE points always refer to the NE of the D2D non-cooperative game $G_D$.
In this subsection, we study what these NE points are and propose a algorithm that converges to a NE.

We denote the feasible region of  $x_{i}$ by $\mathcal{X}_i$, where $\mathcal{X}_i = \left\{x_i \in [0,1] \right\}$. 
The existence of NE for the non-cooperative game is given by Lemma~\ref{lemma:existD2DNE}, according to the Debreu-Glicksberg-Fan Theorem~\cite{Deb52,Fan52,Gli52}.
\begin{lemma}\label{lemma:existD2DNE}
If $\mathcal{X}_i$ is compact and convex, $U_i$ is concave in $x_i$ given $\mathbf{x}_{-i}$ and  continuous, then the NE exists.
\end{lemma}
It is straightforward to show that the above  conditions are satisfied, and thus we have at least one NE. Then a natural question follows: how to attain a NE? 

For fixed $\mathbf{x}_{-i}$ and $\mu$, the problem (\ref{eq:D2Dopt}) is a convex optimization, and the optimal solution is the point which vanishes the first derivative of the objective function (if feasible): 
\begin{equation}\label{eq:bestx}
\begin{aligned}
x_i^* &= \left[\frac{w_i}{\mu P_{D_i}g_{ii}\ln 2}  \right. \\
& \left.  -\frac{1}{\sum_{\mathcal{I}\in  2^\mathcal{D}} \prod_{j\in\mathcal{I}\setminus i} x_j \prod_{n\in\mathcal{D}\setminus\mathcal{I}} (1-x_n)\sinr_{i}^{(D)}}\right]_0^1,
\end{aligned}
\end{equation}
where $[x]_0^1=\min\{1, \max\{0, x\}\}$. We define the following function: $f(x_{1},\cdots, x_{N_D}; \mu)=\left(x^*_1(\mathbf{x}_{-1}),\cdots,x^*_{N_D}(\mathbf{x}_{-N_D})\right)$, where $x_i^*(\mathbf{x}_{-i})$ is given by~(\ref{eq:bestx}). Function $f$ describes the optimal resource access probabilities given that the access probabilities of other links are fixed, and thus is  called the \emph{best-response (BR) function}. We propose a synchronous iterative algorithm -- called the \emph{BR Algorithm}, where all D2D links adjust their access probabilities simultaneously according to 
\begin{equation*}
\left(x_1(t+1),\cdots,x_{N_D}(t+1)\right) = f(x_{1}(t),\cdots, x_{N_D}(t);\mu).
\end{equation*}
Applying the \emph{Maximum Theorem} \cite{ShuLeu07}, we can show that $f$ is continuous. Note that the BR Algorithm will never converge to a solution that is not a NE, since each D2D link has the access probability that maximizes its utility, which implies that no links can gain by changing only their own access probabilities unilaterally at the convergence point.

Though procedures of the BR Algorithm are simple, the complexity to calculate (\ref{eq:bestx}) is high, due to the expectation calculation involving $N_D$ Bernoulli random variables, whose complexity is $O(2^{N_D} N_D^2)$. In addition,  D2D links need to exchange their current access probabilities, which causes high overhead. The overhead and complex computation are not desirable, especially for UEs that are power limited. Other algorithms such as gradient-projection based algorithm \cite{BerTsi89} or algorithms in learning automata \cite{NarTha12} can also be applied, with the disadvantages of either slow convergence or memory space limit. These motivate the following subsection, where we consider  a lower bound of the objective function in~(\ref{eq:D2Dopt}).

\subsection{Joint Resource Allocation and Power Control -- A Lower Bound Problem}\label{sec:D2DgameLB}
In problem (\ref{eq:D2Dopt}), each D2D link maximizes the utility in terms of the expected SINR. Approximating the rate to be calculated by expected interference rather than expected SINR, we have
\begin{equation}\label{eq:opt-D2D-cvx}
\begin{aligned}
\max_{x_i} \ & w_i  \log\left(1+\sinr'_i \right) - \mu x_{i}P_{D_i}g_{ii}\\
\text{s.t. } & 0\leq x_{i} \leq 1,
\end{aligned}
\end{equation}
where $\sinr'_i=\frac{x_{i}P_{D_i}h_{ii}}{\sum_{j\in\mathcal{D},j\neq i}x_{j}P_{D_j}h_{ji} +\sum_{j\in\mathcal{C}} P_{C_j}h_{ji}+W_{i}}$.
This problem motivates a low-complexity low-overhead algorithm, as shown below.

Variables $x_i$ in~(\ref{eq:opt-D2D-cvx}) can  be considered as a joint resource allocation and  power control variable, where $\mathbbm{1}\left(x_i>0\right)$ indicates whether D2D link $i$ accesses the RB, and the value of $x_i$ denotes the fraction of maximal transmit power to use.  
The strategy with respect to (\ref{eq:D2Dopt}) can be considered as a scheme similar to random hopping (with different hopping probabilities at each link), while the strategy in (\ref{eq:opt-D2D-cvx}) is deterministic, which considers power control in addition to resource allocation. Intuitively, the hopping scheme  randomizes  strong interference, and thus may potentially provide a larger gain than the latter case, though we consider power control jointly. We show this relationship mathematically in Proposition \ref{prop:LBD2D}.

\begin{prop}\label{prop:LBD2D}
The optimization problem (\ref{eq:opt-D2D-cvx}) maximizes a lower bound of the utility function in~(\ref{eq:D2Dopt}).
\end{prop}
\begin{proof}
Denoting the interference from other D2D links by $I$, the SINR can be written as
$
\mathbb{E}\left[\sinr_{D_i}\right] = \mathbb{E}_I \left[\frac{P_{D_i}h_{ii}}{I + \sum_{j\in\mathcal{C}} P_{C_j}h_{ji}+W_{i}}  \right].
$
It is straightforward to verify that $f(I) =\frac{P_{D_i}h_{ii}}{I + \sum_{j\in\mathcal{C}} P_{C_j}h_{ji}+W_{i}} $  is convex. By Jensen's inequality, we have $f(\mathbb{E}[I]) \leq \mathbb{E}\left[f(I)\right]$, which completes the proof.
\end{proof}

We call (\ref{eq:opt-D2D-cvx}) the lower bound problem of (\ref{eq:D2Dopt}) in this paper. Invoking Lemma \ref{lemma:existD2DNE} again, we can show that there is at least one NE for the D2D game formulated in this subsection.  Though there may exist multiple NEs in general, our setup admits a unique NE under some specific conditions; we specify those precisely in Section \ref{sec:iterative-d2dalgo}. Note that the NEs of the games with (\ref{eq:D2Dopt}) and with (\ref{eq:opt-D2D-cvx}) are not necessarily the same, and thus Proposition~\ref{prop:LBD2D} does not say that the BR Algorithm in Section \ref{sec:D2Dgame1} always performs better than the algorithms proposed in the following subsection.

Given $\mathbf{x}_{-i}$ and $\mu$, (\ref{eq:opt-D2D-cvx}) is a convex optimization problem and its optimal solution is given by Proposition \ref{prop:optimalD2D}.
\begin{prop}\label{prop:optimalD2D}
The solution of (\ref{eq:opt-D2D-cvx}) has the following form
\begin{equation}\label{eq:optsol-D2D}
x_{i} ^*=   \left[\frac{a_{i}-s_{i}}{P_{D_i}h_{ii}}\right]_0^1,
\end{equation}
where $a_{i} = \frac{w_iP_{D_i}h_{ii}}{\mu P_{D_i} g_{ii}} - \sum_{j\in\mathcal{C}} P_{C_j}h_{ji}-W_{i}$, $s_{i}=\sum_{j\neq i} x_{j}P_{D_j}h_{ji}$, and $[x]_0^1= \min\{1, \max\{x, 0\}\}$.
\end{prop}
\begin{proof}
According to the  Karush-Kuhn-Tucker (KKT) conditions~\cite{BoyVan04}, we have $\frac{\partial U_i}{\partial x_{i}}=0$ if $x_i\in(0,1)$, $\frac{\partial U_i}{\partial x_{i}}\leq 0$ if $x_i=0$, and $\frac{\partial U_i}{\partial x_{i}}\geq 0$ otherwise, where
$\frac{\partial U_i}{\partial x_{i}}=\frac{w_iP_{D_i}h_{ii}}{x_{i}P_{D_i}h_{ii} + \sum_{j\in\mathcal{D},j\neq i}x_{j}P_{D_j}h_{ji}+\sum_{j\in\mathcal{C}} P_{C_j}h_{ji}+W_{i}} - \mu_{i}P_{D_i}g_{ii}$.
The above equations and inequations result in (\ref{eq:optsol-D2D}).
\end{proof}

Eq.  (\ref{eq:optsol-D2D}) is similar to the waterfilling function in power allocation problems, except that our constraint $x_i\in[0,1]$ is independent over different RBs and thus we obtain a closed-form solution  (\ref{eq:optsol-D2D}). Leveraging existing works on waterfilling problems, we propose an iterative algorithm similar to the iterative waterfilling algorithm (see, e.g., \cite{YuCio02,JinRhe05,ShuLeu07}).

\subsection{Algorithm Design for the Lower bound Problem}\label{sec:iterative-d2dalgo}
Similar to Section \ref{sec:D2Dgame}, we propose a synchronous iterative algorithm based on the BR function, defined as $f_L(x_{1},\cdots, x_{N_D}; \mu)=\left(x^*_1(\mathbf{x}_{-1}),\cdots,x^*_{N_D}(\mathbf{x}_{-N_D})\right)$, where $x_i^*(\mathbf{x}_{-i})$ is given by~(\ref{eq:optsol-D2D}). The algorithm -- called the \emph{LB Algorithm} -- is given by Algorithm~\ref{algo:d2d}. Similar to the BR Algorithm, we have that if the LB Algorithm converges, then it converges to a NE. 
\begin{algorithm}
  \caption{LB Algorithm: an iterative  algorithm for lower bound problem of D2D}\label{algo:d2d}
  \begin{algorithmic}[1]
    \State Initialization: given price $\mu \geq 0$, let $x_i(0)=1,\forall i$, and $t=0$;
    
    \While{ $\|\mathbf{x}(t) - \mathbf{x}(t-1)\|\geq \epsilon$ } 
    \State let $\mathbf{x}(t+1) = f_L(\mathbf{x}(t); \boldsymbol{\mu})$;
    \State let $t=t+1$;
    \EndWhile  
  \end{algorithmic}
\end{algorithm}

\textbf{Implementation interpretations.} Adopting the LB Algorithm, each D2D link first acquires CSI of the link from its transmitter to the BS. This can be either estimated based on the downlink signal (e.g., in a  time-division duplexing (TDD) uplink/downlink configuration), or provided by the BS, which measures the uplink channel and sends the information to the D2D user (e.g., in a frequency-division duplexing (FDD) uplink/downlink configuration). Apart from uplink CSI, each D2D link also measures the channel between the transmitter to its paired receiver. The frequency to update CSI depends on the channel variance. For example, in a slow mobility scenario,  D2D links may just update the information once (at the beginning of each resource allocation period). At each iteration of the LB Algorithm, every D2D link measures the interference if accessing a RB. There is no additional message exchange in this step. Thus,  the LB Algorithm only requires local information and reduces the overhead, as shown in~Fig.~\ref{fig:algostructure}.

\begin{figure}
\centering
\includegraphics[width=8cm, height=6.2cm]{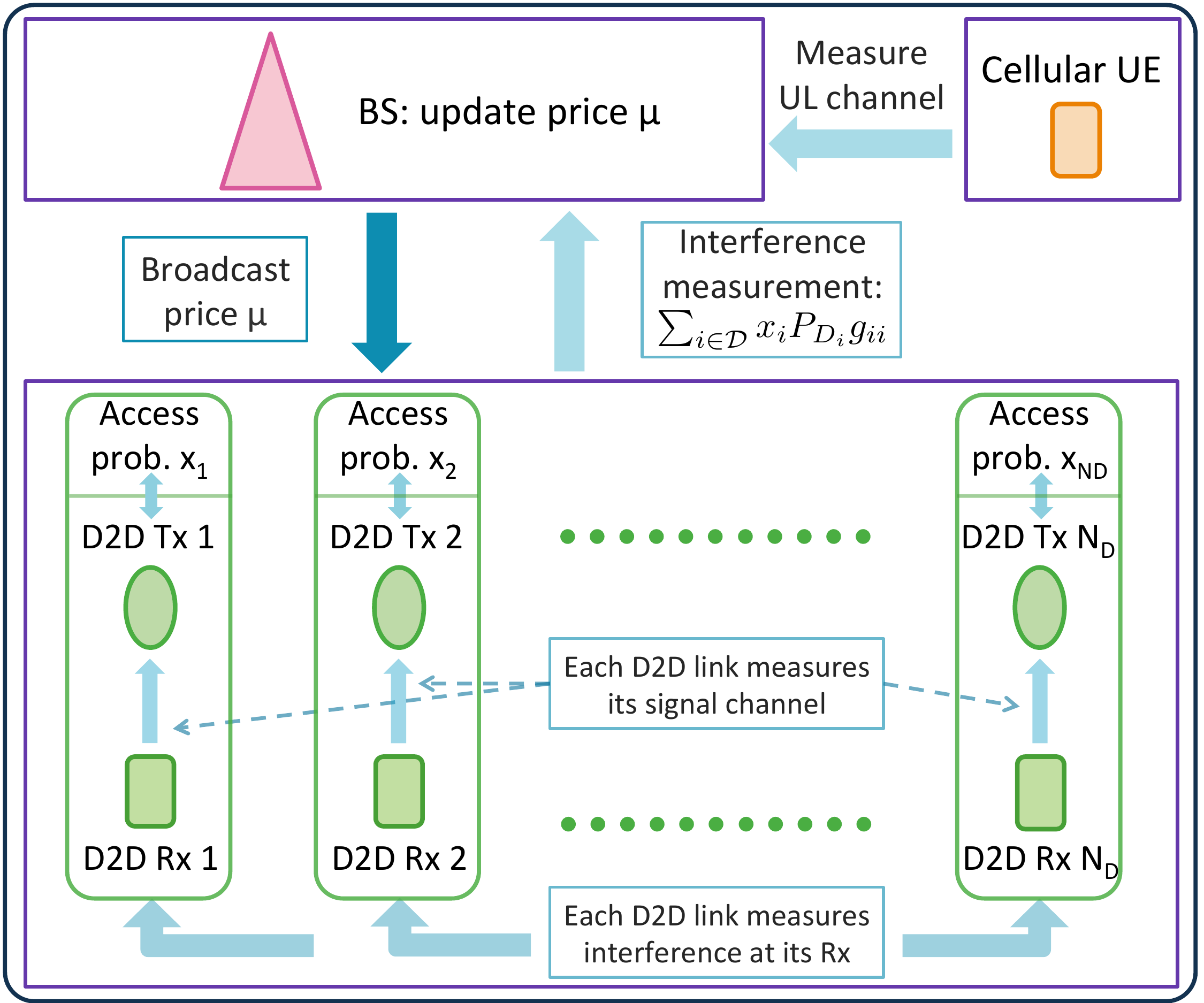}
\caption{Illustration of the proposed algorithm. The arrows filled with dark color indicate the procedures requiring message exchange, while the arrows filled with light color indicate the procedures involving only local measurements. The lower part describes the LB Algorithm for the lower problem, while the upper part illustrates algorithms proposed for the upper problem.}
\label{fig:algostructure}
\end{figure}

\textbf{Convergence analysis.} To get the convergence criteria of the LB Algorithm, we first investigate some basic properties of function $f_L$. We assume that there are finite number of D2D links. We call the set of D2D links with $x_{i}=1$ saturated D2D links, denoted by $\mathcal{S}\defeq \left\{i\in\mathcal{D}: x_{i}=1\right\}$, and the set of D2D links with  $x_{i}\in(0,1)$ active D2D links, denoted by $\mathcal{A}\defeq \left\{i\in\mathcal{D}: x_{i}\in(0,1)\right\}$.  Denoting $\mathbf{s} =[s_1, \cdots, s_{N_D}]^T$, we have $\mathbf{s}=\mathbf{G}\mathbf{x}$, where $\mathbf{G}$ is an $N_D \times N_D$ matrix with zero diagonal elements and $(i,j)$th element (with $i\neq j$) being~$P_{D_j}h_{ji}$. 

\begin{prop}\label{prop:BRfunc-property}
The best-response function $f_L$ has the following properties:
\begin{enumerate}
\item $f_L$ is a continuous mapping from $\mathcal{X}$ to $\mathcal{X}$.
\item $f_L$ is piecewise affine, which means that $f_L$ has the following two properties:
\begin{enumerate}
\item  The domain of function $f_L$ can be partitioned into finitely many polyhedral regions, denoted by $\mathcal{P}_1,\cdots,\mathcal{P}_d$, which are determined by the choice of $\mathcal{A}$ and $\mathcal{S}$;
\item On the polyhedron $\mathcal{P}_n$ defined by $\mathcal{A}^{(n)}$ and $\mathcal{S}^{(n)}$, we have $f_L(\mathbf{x})=\mathbf{M}^{(n)}\mathbf{x+b}^{(n)}$, where $\mathbf{b}^{(n)}$ is a constant vector, and $\mathbf{M}^{(n)}=\mathbf{B}^{(n)}\mathbf{G}$ with $\mathbf{B}^{(n)}$ being a diagonal matrix, which has $\left[\mathbf{B}_{i}^{(n)}\right]_{kl} =-\frac{1}{P_{D_i}h_{ii}}$ if $k = l, i \in\mathcal{A}^{(n)}$, and $\left[\mathbf{B}_{i}^{(n)}\right]_{kl}=0$ otherwise.
\end{enumerate} 
\end{enumerate}
\end{prop}
\begin{proof}
See Appendix \ref{pf:prop-BRfunc-property}.
\end{proof}
We assume that the resource allocation is carried out well during the channel coherence time, and thus channel can be regarded as static during resource allocation updates. We leave the stochastic channel analysis as future work.  Defining the matrix norm of a  matrix~$\mathbf{M}$ induced by the vector norm $\| \cdot\|$ as $\|\mathbf{M}\|\defeq \max\left\{\|\mathbf{Mx}\|: \|\mathbf{x}\|=1 \right\}$ \cite{ShuLeu07}, a sufficient condition for the convergence of the proposed algorithm with general matrix norms can be found as follows, leveraging the techniques used  in  Theorem 7 in \cite{ShuLeu07}.

\begin{theo}\label{theo:f-contractmap}
If $\|\mathbf{M}_n\|<1$, we have
\begin{enumerate}
\item the synchronous iterative algorithm converges for any initial resource allocation;
\item there is a unique fixed point $\mathbf{x}^*$;
\item $\|\mathbf{x}{(t)} -\mathbf{x}^*\| \leq \eta^t \|\mathbf{x}{(0)} -\mathbf{x}^*\|$, where $\eta = \max_n \|\mathbf{M}_n\|$. The upper bound of the convergence rate is~$\eta$.
\end{enumerate}
\end{theo}
\begin{proof}
If $f_L$ is a contraction mapping, 
then global stability follows from the Banach Fixed Point Theorem (see, e.g.\cite{Gra03}). The proof of contraction mapping is similar to the proof of Theorem 7 in~\cite{ShuLeu07}, and thus we ignore the details.  Given that $f_L$ is a contraction mapping, we have 
$
|f_L(\mathbf{x}') - f_L(\mathbf{x})\| \leq  \eta \|\mathbf{x}-\mathbf{x}'\|.
$ 
 The rate of convergence for a sequence $\{x_n\}$ converging to $L$ is defined as the $\lim_{n\rightarrow \infty} \frac{|x_{n+1}-L|}{|x_n-L|}$. Observing the above  inequality, we conclude that the convergence rate of the BR Algorithm is upper bounded by~$\eta$.
\end{proof}
The number of polyhedral regions that partition  the domain of $f_L$ is $O(3^{N_D})$, which is very large, and thus it is impractical to check the conditions in Theorem \ref{theo:f-contractmap} directly for all regions. We further provide sufficient conditions in Proposition \ref{prop:f-convergecond} that are easy to apply.
\begin{prop}\label{prop:f-convergecond}
If the matrix $\mathbf{G}$ satisfies
$
\|G\| \leq \min_{i,k} P_{D_i}h_{ii}^{(k)},
$
then the algorithm converges to the unique fixed point regardless of the initial point.
\end{prop}
\begin{proof}
According to Prop. \ref{prop:BRfunc-property}, we have $\mathbf{M}_n=\mathbf{B}^{(n)}\mathbf{G}$. To make $f_L$  a contraction mapping, we have to satisfy $\|\mathbf{B}^{(n)}\mathbf{G}\|\leq 1$. By the property of matrix norm that $\|\mathbf{AB}\|\leq \|\mathbf{A}\|\cdot \|\mathbf{B}\|$, we obtain a sufficient condition that is $\|\mathbf{B}^{(n)}\|\cdot \|\mathbf{G}\|\leq 1$. Matrix $\mathbf{B}^{(n)}$ is a diagonal matrix, whose norm is
$
\|\mathbf{B}^{(n)}\| \leq \max_{i} \frac{1}{P_{D_i}h_{ii}},\ \forall n.
$
Then we can get one sufficient condition as $\|G\| \leq \left( \max_{i,k} \frac{1}{P_{D_i}h_{ii}^{(k)}}\right)^{-1} = \min_{i,k} P_{D_i}h_{ii}^{(k)}$.
\end{proof}
\textbf{Design interpretations.} The above result is true for any general $l_p$ norm with $p\geq 1$. As in \cite{ShuLeu07}, we apply it to some special matrix norms and give the corresponding interpretations as follows.

\noindent\textit{Example 1 ($l_1$ norm)}. We have $\|\mathbf{G}\|_1 = \max\left\{ \|\mathbf{Gx}\|_1: \|\mathbf{x}\|_1=1\right\}=\max\left\{\sum_{i=1} |s_i|\right\}$. This implies that a sufficient condition for the convergence of the LB Algorithm is that no D2D transmitter causes very strong interference to other D2D links.

\noindent\textit{Example 2 ($l_\infty$ norm)}.  We have $\|\mathbf{G}\|_\infty = \max\left\{ \|\mathbf{Gx}\|_\infty :\|\mathbf{x}\|_\infty=1\right\}=\max\left\{\max_i |s_i|\right\}$. This  implies that a sufficient condition for the convergence  is that no D2D receiver suffers excessive interference. 

\begin{figure}
\centering
		\subfigure[$\mu=80$dB.]{
			\label{fig:location80}
			\includegraphics[height=5.9cm]{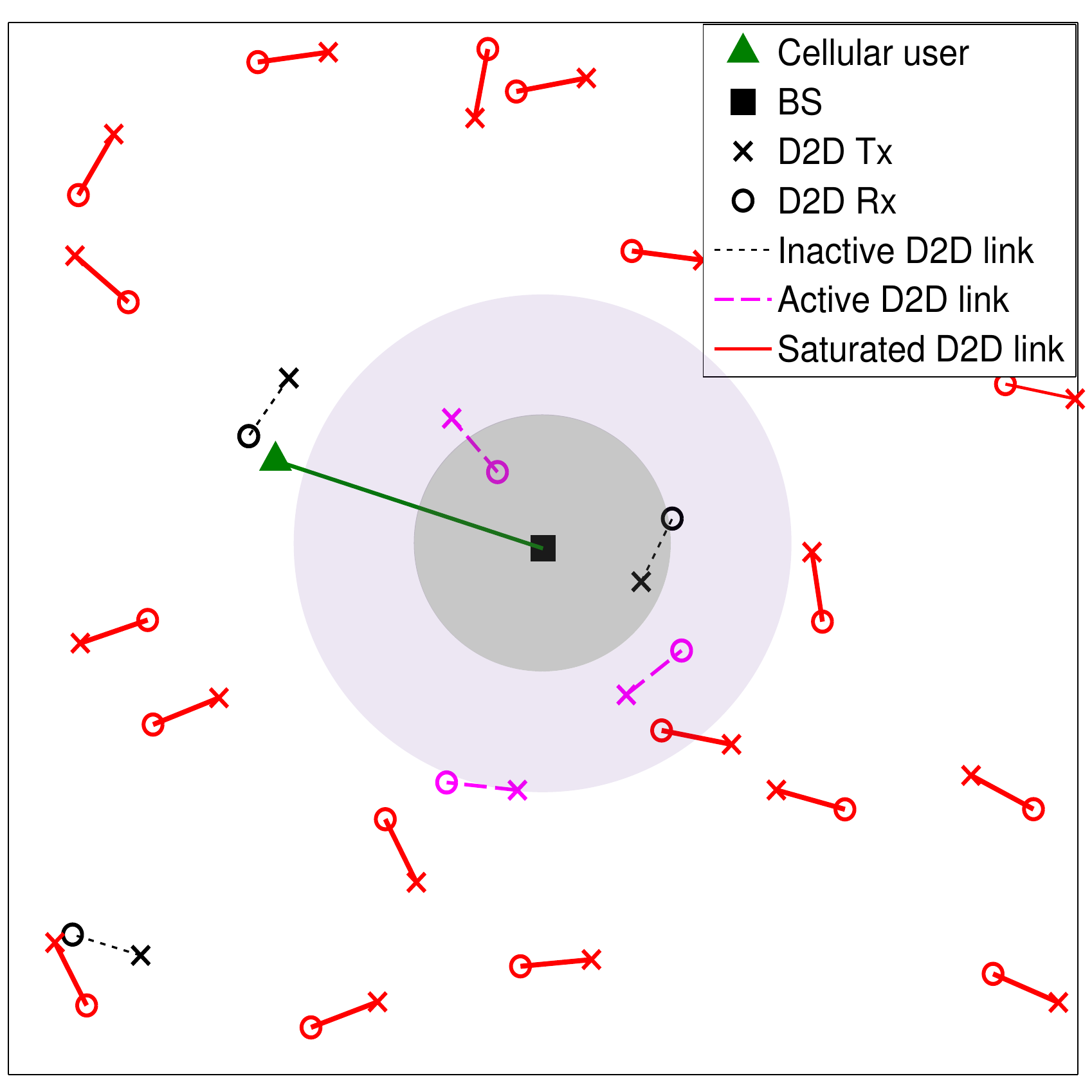}}
		\subfigure[$\mu=90$dB.]{
			\label{fig:location90}
			\includegraphics[height=6cm]{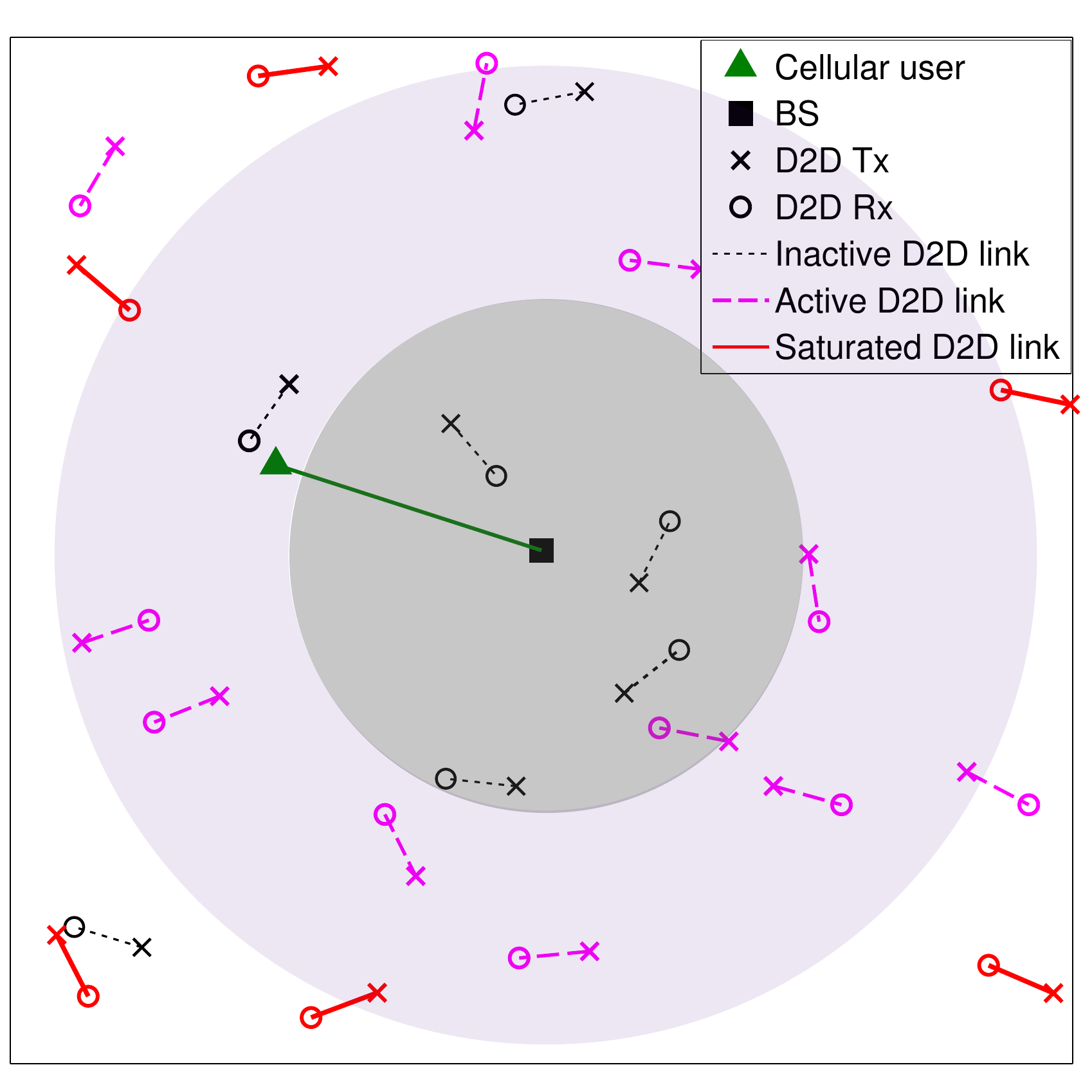}}
\caption{The access probabilities of D2D links vs. $\mu$. The areas in the dark shade  show the locations of silent D2D links with $x_i=0$. The light shaded areas  show the locations of active D2D links (i.e., $x_i\in(0,1)$). The remaining parts show the locations of saturated D2D links (i.e., $x_i=1$).}\label{fig:location}
\end{figure}
We show examples of D2D access probabilities obtained by the BR Algorithm versus different $\mu$  in Fig. \ref{fig:location}. We can observe that the BR Algorithm discourages D2D links whose transmitters are near the BS from accessing the RB. 
In addition, the BR Algorithm takes  into account the SINR of D2D links, and encourages a D2D link far from the BS to keep silent if there are many D2D links nearby,  to decrease the interference in D2D networks. Comparing the two subfigures, we conclude that  less D2D links would be active as $\mu$ increases. This suggests the potential effectiveness of $\mu$, which is investigated in the following section.

\section{Upper Problem: Network's Pricing Mechanism}\label{sec:uppergame}
It is difficult to analyze the upper problem by plugging (\ref{eq:bestx}) directly into (\ref{eq:opt-BS}), due to the complex term in the denominator. Simulation results in Section \ref{sec:simulation} show that the performance of the LB Algorithm for problem (\ref{eq:opt-D2D-cvx}) is very close to the performance of the BR Algorithm for the original problem (\ref{eq:D2Dopt}). This suggests to  approximate the solution of  (\ref{eq:D2Dopt}) to the solution of the lower bound problem (\ref{eq:optsol-D2D}). By this approximation, we propose an algorithm leveraging techniques from LCP \cite{CotPang92}. We further propose a bisection algorithm, which has low overhead and can be applied to the original two-stage problem (with the lower problem~(\ref{eq:D2Dopt})). In addition, motivated by the results of the lower problem (see e.g., Fig.~\ref{fig:location}), we propose a simple greedy heuristic algorithm, which can be applied to the network with high interference tolerance level.

\subsection{An Equivalent Upper Problem}
According to Lemma~1 in~\cite{RazLuo11},  the upper problem is equivalent to the following problem: 
\begin{equation}\label{eq:subproblem}
\begin{aligned}
\max_{\mu \geq 0} \ & \min\left\{\mu \sum_{i\in\mathcal{D}}x_{i}P_{D_i}g_{ii}, \ \mu Q\right\}\\
\text{s.t. } & x_{i} =  x_i^*,
\end{aligned}
\end{equation}
where $x_i^*$ is given by (\ref{eq:bestx}).
 For simplicity, we use the notation $0\leq a \perp b \geq 0$ to represent the complementarity condition of $a$ and $b$, i.e., $ab=0$ and $a,b\geq0$ \cite{CotPang92}. Letting $I_{Ci} = \sum_{j\in\mathcal{C}} P_{C_j} h_{ji} + W_{i}$, and  $\lambda_i$ be the Lagrange multiplier to relax the constraint $x_i \leq 1$, we have the following lemma.
\begin{lemma}
Denoting $t_i = \frac{w_iP_{D_i}h_{ii}}{u P_{D_i}g_{ii}\left(uP_{D_i}g_{ii} +\lambda_i\right)}$, (\ref{eq:optsol-D2D}) is equivalent to the following parametric LCP  with variables $(x_i, t_i)$ and  parameter $\mu$~\cite{CotPang92,RazLuo11}:
\begin{equation}\label{eq:LCP}
\begin{aligned}
& 0\leq  x_{i} \perp \left(t_i - \frac{w_i h_{ii} }{\mu g_{ii}}+ \sum_{j\in\mathcal{D}} P_{D_j} h_{ji}x_j +I_{Ci}\right)\geq 0,\\
& 0 \leq t_i \perp (1-x_i) \geq 0.
\end{aligned}
\end{equation}
\end{lemma}
\begin{proof}
Details can be found in \cite{RazLuo11}. Key steps are to multiply (\ref{eq:optsol-D2D}) by $\frac{\sum_{j\in\mathcal{D}} P_{D_j} h_{ji}x_j +I_{Ci}}{\mu P_{D_i}g_{ii} + \lambda_i}$ and change the variable $t_i$.
\end{proof}
 
In the following, we explore properties of the objective function in (\ref{eq:subproblem}), which provides clues to design efficient algorithms for the upper problem.

\subsection{Algorithm Design for the Upper Problem}
In this paper, we use the symmetric parametric principle pivoting algorithm (SPPP) -- a classical algorithm for parametric LCP  \cite{CotPang92} -- to find the optimal $\mu$ and its corresponding feasible solutions $x_i$ in (\ref{eq:LCP}).  
We write (\ref{eq:LCP}) in matrix form as $\mathbf{0}\leq \mathbf{y}\perp \mathbf{Ay + q} + \nu\mathbf{d}\geq 0$,
where $\nu=\frac{ 1}{\mu}$, $\mathbf{y} = [x_1, \dots, x_{N_D}, t_1, \dots, t_{N_D} ]^T$, $\mathbf{q} = [I_{C1} \dots, I_{CN_D} , 1, \dots, 1]^T$, $\mathbf{d} =[-\frac{w_1 h_{11}}{ g_{11}},\dots, - \frac{w_{N_D} h_{N_DN_D}}{ g_{N_DN_D}}, 0, \dots, 0]^T$,  $\mathbf{A}_0$ is a matrix with $(i,j)$th element $P_{D_j}h_{ji}$, and $\mathbf{A}=\bigl[\begin{smallmatrix}
\mathbf{A}_0 & \quad \mathbf{I}\\
-\mathbf{I}  &\quad  \mathbf{0}
\end{smallmatrix} \bigr]$.
Note that $\mu\geq 0$ implies that $\nu \geq 0$. We set a upper limit for $\nu$, denoted by $\bar{\nu}<\infty$, which is a sufficient large real number. The SPPP is given by Algorithm \ref{algo:SPPP}.
\begin{algorithm}
  \caption{SPPP-based algortihm}\label{algo:SPPP}
  \begin{algorithmic}[1]
    \State Initialization:  $\tau=0$,  $\mu^*=0$,  $U^*=0$, $\mathbf{q}(\tau)=\mathbf{q}$, $\mathbf{d}(\tau) = \mathbf{d}$, $\mathbf{A}(\tau) = \mathbf{A}$, $\nu(\tau) = 0$, $\mathbf{y}(\tau)=\mathbf{1}$, and
    
     \quad \quad \quad \ \  $\mathbf{z}(\tau) = \mathbf{q+\nu(\tau) d + Ay(\tau)}$;
    
\Statex $\triangleright$ {\textbf{comment}: critical value}
        
    \State Determine the next critical value of $\mu$: 
    \begin{equation*}
    \nu(\tau+1) =\min\left\{ \min_i \left\{-\frac{q_i(\tau)}{d_i(\tau)}:d_i(\tau) < 0\right\}, \bar{\nu} \right\};
    \end{equation*}

    \State Set $\left(\mathbf{z}(\tau+1), \mathbf{y}(\tau+1)\right) =  \left(\mathbf{q}(\tau) + \nu\mathbf{d}(\tau), 0 \right)$ for all $\nu\in[\nu(\tau), \nu(\tau+1) ]$;   
     \If{ $\nu(\tau+1)=\bar{\nu}$}	
     		\State stop;
     \Else
     		\State let $r=\arg\min_i \left\{ -\frac{q_i(\tau)}{d_i(\tau)}:d_i(\tau) < 0\right\}$;
     \EndIf
   \State The new critical value of $\nu$ is $\nu(\tau+1) = - q_r(\tau)/d_r(\tau) $, and thus $\mu(\tau + 1) = (\nu(\tau+1))^{-1}$;
    
    \Statex $\triangleright$ {\textbf{comment}: pivoting}

   \If{$A_{rr}(\tau)>0}$
   		\State pivot $<z_r(\tau), y_r(\tau)>$;
   		\State let $z_r(\tau+1) = y_r(\tau)$,  $y_r(\tau+1) = z_r(\tau)$;
   		\State let  $z_i(\tau+1) = z_i(\tau)$, $y_i(\tau+1) = y_i(\tau)$, for $i\neq r$;
   		\State let $\tau = \tau + 1$, and return to Step 2;
	
	\ElsIf {$A_rr(\tau)=0$ }
		\State use $y_r(\tau)$ as a driving variable and determin the basic blocking variable $z_s(\tau)$;
		\State pivot  $<z_s(\tau), y_r(\tau)>$, $<z_r(\tau), y_s(\tau)>$;
   		\State let $z_s(\tau+1) = y_r(\tau)$,  $y_s(\tau+1) = z_r(\tau)$, $z_r(\tau+1) = y_s(\tau)$,  $y_r(\tau+1) = z_s(\tau)$;
   		\State let $z_i(\tau+1) = z_i(\tau)$, $y_i(\tau+1) = y_i(\tau)$, for $i\neq r, s$;
   		\State let $\tau = \tau + 1$, and return to Step 2;
     \EndIf

    \State get $x_i(\tau+1)$ from $y_i(\tau+1)$;
    \State let $U$ be (\ref{eq:subproblem}) at $\mu(\tau+1)$ and $x_i(\tau+1)$;
    \If {$U > U^*$}
    	\State let $U^* = U$, $\mu^*=\mu(\tau+1)$;
   	 \EndIf  
    
  \end{algorithmic}
\end{algorithm}

Since the SPPP Algorithm requires CSI between each transmitter and receiver (i.e., matrix $\mathbf{A}$), which may cause high overhead,  the result of SPPP can be used as a performance benchmark, and another low-overhead algorithm is desirable. To propose such algorithms, we first explore the properties of the objective function in (\ref{eq:subproblem}), which is denoted by $U_c \defeq \min\{ U_{c1}, U_{c2} \}$, where $U_{c1} =\mu \sum_{i\in\mathcal{D}}x_{i}P_{D_i}g_{ii}$ and $U_{c2}=  \mu Q$. Function $U_{c2}$ is a linear increasing function of $\mu$, while  $U_{c1}$ is more complicated since it involves solving (\ref{eq:optsol-D2D}). The properties of function $U_{c1}$ are given by Proposition \ref{prop:U1-property}, leveraging the techniques in \cite{RazLuo11}.
\begin{prop}\label{prop:U1-property}
The function $U_{c1}(\mu)$  has the following properties:
\begin{enumerate}
\item $U_{c1}$ is a continuous function of $\mu$;
\item $U_{c1}$ is piecewise affine;
\item  If  $\sum_{j\in\mathcal{D}, j\neq i} \frac{ h_{ij}}{h_{jj}} g_{jj} < g_{ii},\ \forall i\in\mathcal{D}$, and $\sum_{j\in\mathcal{D}_s}P_{D_j}\left(h_{ji}- \frac{h_{1i}}{g_{11}} g_{jj}\right)\geq 0, \ \forall i\in\mathcal{D}_a$, 
then $U_{c1}$ is a non-increasing function.
\end{enumerate}
\end{prop}
\begin{proof}
See Appendix \ref{pf:prop-U1-property}.
\end{proof}

The sufficient conditions given in Proposition \ref{prop:U1-property} to make $U_{c1}$ non-decreasing essentially say that the interference among D2D links and the interference from saturated D2D links (i.e., $x_i=1$) to the BS should be weak. Given by Proposition \ref{prop:U1-property} that $U_{c1}$ is piecewise affine, and $U_{c2}$ is linear, the optimal $\mu^*$ must either be at a \emph{break point} of $U_{c1}$ -- the discontinuous points in the derivative of $U_{c1}$ --  or at the intersection of $U_{c1}$ and $U_{c2}$ \cite{RazLuo11}. When $U_{c1}$ is non-decreasing, the optimal $\mu^*$ must be at the intersection of  $U_{c1}$ and $U_{c2}$. This motivates the heuristic algorithm, given by Algorithm~\ref{algo:mono-U1}, which converges to an intersection point of  $U_{c1}$ and $U_{c2}$ \cite{RazLuo11}. 

\begin{algorithm}
  \caption{A bisection algorithm for finding optimal price $\mu^*$ at monotonic case}\label{algo:mono-U1}
  \begin{algorithmic}[1]
    \State Initialization: given accuracy $\epsilon\geq 0$, let $\mu_u=\mu_{\max}$, and $\mu_l\geq 0$;
    
    \While{ $|\mu_u-\mu_l|\geq \epsilon$ } 
    \State let $\mu_m = \frac{\mu_u + \mu_l}{2}$;
    \State get $x_i(\mu_m)$ by running the LB Algorithm;
    \If{  $U_1(\mu_m, x_i(\mu_m)) \leq U_2(\mu_m, x_i(\mu_m))$ }
    \State $\mu_u=\mu_m$;
    \Else
    \State $\mu_l=\mu_m$;
    \EndIf  
    \EndWhile  
    \State let $\mu^*=\mu_m$.
  \end{algorithmic}
\end{algorithm}
In the bisection algorithm, we set an upper limit for $\mu$, denoted by $\mu_{\max}<\infty$, which is a sufficient large real number. We consider non-trivial cases, where the interference from D2D to BSs, when all D2D links are active, is greater than the interference tolerance level; otherwise, we just let all D2D links access a RB with probability one. Under this assumption, we have the following result.
\begin{prop}\label{prop:bisection}
The bisection algorithm always converges. In particular, the algorithm requires at most $\log_2(\mu_{\max}/\epsilon)$ iterations to converge. 
\end{prop}
\begin{proof}
See Appendix \ref{pf:prop-bisection}.
\end{proof}
Note that the bisection algorithm also converges when we use the BR Algorithm instead of the LB Algorithm to solve the lower problem, due to that  function $f$ is continuous.  Under the conditions given by Proposition \ref{prop:U1-property}, i.e., the interference among D2D links and the interference from saturated D2D links to BSs are weak, the bisection algorithm achieves the optimal $\mu^*$. In other words, the optimal strategy in this case is to let the number of active D2D links as large as possible, until the total interference from D2D links reaches the tolerance~level.

\textbf{Implementation interpretations.} Adopting the bisection algorithm, the BS first broadcasts a price, and then measures the aggregate interference at this price. If the interference is greater than the tolerance level, the BS increases the price; otherwise, the BS  decreases the price. In fact, the behavior is consistent with the \emph{law of supply and demand}: if the demand (the interference) exceeds the supply (the interference tolerance level), the price increases to make the RB less attractive. The algorithm can also be implemented adaptively. The network locally measures the total D2D interference, and increases (decreases) the price if the interference level is above (below) the predefined tolerance level $Q$, until the interference level reaches $Q$. 
Fig.~\ref{fig:algostructure} illustrates the  structure  of Algorithm~\ref{algo:mono-U1}, which shows that the signaling overhead is caused by the price broadcast and the channel measurements. The overhead due to price broadcast is proportional to the number of RBs, which is quite small. As for the channel measurements, the BS requires the channel information of cellular links,  and each D2D link needs the CSI of  the link between its transmitter to the paired receiver and of the link between the transmitter and the BS. Thus, the algorithm only requires local information and the overhead due to the channel measurements is proportional to the total number of cellular and D2D links, which is much lower than the overhead of centralized algorithms (e.g., the brute force approach or the SPPP Algorithm) that require  global CSI. Note that the channel information updating frequency depends on the channel variance over time, which is quite low in a slow mobility environment (e.g., we may only measure channels once for each or several resource allocation time scales). Therefore, the required overhead is not significant compared to the potential advantages of our algorithm.

As for the complexity, letting $T$ be the number of required iterations for the LB Algorithm, the computational complexity of Algorithm~3 is $O(N_D\log_2(\mu_{\max} / \epsilon) + N_D^2 T)$. The parameters $T$ and $\log_2(\mu_{\max} / \epsilon)$ are generally much smaller  than $N_x^{N_D}$  as illustrated in Section~\ref{sec:simulation}, where $T$ and $\log_2(\mu_{\max} / \epsilon)$ are between 5 and 10, while $N_x^{N_D}$ is $10^{10}$.  Thus, the complexity of Algorithm 3 is much lower than the complexity of the centralized scheme, which is $O(N_x^{N_D} N_D^2)$.

Observing Fig. \ref{fig:location}, D2D links mostly have larger access probabilities when they are far from the BS. This motivates another greedy heuristic algorithm -- called the \emph{IO Algorithm} (short for interference ordering), which needs no iteration. The D2D links are sorted by the interference caused to the BS in an ascending order, i.e., $P_{D1}g_{11} \leq P_{D2}g_{22}\leq\cdots\leq P_{DN_D}g_{N_DN_D}$. The BS lets $x_1=1, \dots, x_n=1$ and other D2D links be silent, where $n$ satisfies $\sum_1^n P_{Di}g_{ii}\leq Q$ and $\sum_1^{n+1} P_{Di}g_{ii}> Q$. Adopting the IO Algorithm, the BS measures the uplink CSI from D2D transmitters, based on which the BS determines the access probabilities. Therefore, this algorithm has lower overhead than the bisection algorithm, and gets the solution more quickly, at the cost of overall performance, which is shown in the following section.

\section{Performance Evaluation}\label{sec:simulation}
We consider an uplink system with a hexagonal BS model. The main simulation parameters are listed as follows,  unless otherwise specified. The BS density is 1 per $\pi \left(500\text{m}\right)^2$. The cellular UEs and D2D links are deployed according to two independent Poisson point processes with the same density $10$ links per macrocell. We let the average length of D2D links be $80$m. We assume the total bandwidth is $10$MHz with $1$MHz per sub-band. The transmitters adopt fractional power control, i.e., $P = \min \left\{  P_{\max},  d^{\kappa\alpha}  \right\}$, where $P_{\max}$ is the maximum transmit power, $d$ is the distance of the link, $\kappa$ is the compensation factor for path loss, and $\alpha$ is the path loss exponent. We let cellular UEs and D2D links have the same  power control factor $\kappa=0.75$. The maximum transmit powers of cellular UEs and D2D links over one sub-band are $200$mW and $20$mW, respectively, due to the fact that cellular UEs only access one sub-band while D2D links can access multiple sub-bands. Note that D2D links may not access all sub-bands, and thus we set a conservative maximum transmit power for D2D links. 
The noise  power spectrum density is $-174$ dBm/Hz. Path loss exponents of UE-UE and UE-BS links are~$4.37$ and $3.76$, respectively. We compare the performance of our proposed algorithms to the scenario where all D2D links are active, as well as the scheme where D2D links become silent when their transmitters are within a circle around their nearest BSs -- called a \emph{guard zone scheme}.

\subsection{The Lower Problem: D2D non-cooperative Game}
In this section, we consider a single cell scenario, where the interference tolerance level is $5$dB above the cellular signal. We investigate the performance of the BR Algorithm and the LB Algorithm. 
Note that the BR Algorithm provides the NE result, which may not be optimal. Due to the complexity to solve (\ref{eq:1stage-opt}) via brute force search ($O(N_x^{N_D} N_D^2)$), we compare the gap between the NE and optimal results in a small network with three D2D links. The average total D2D rates obtained by brute force search and by the BR Algorithm are 3.362 bps/Hz and 3.355 bpz/Hz, respectively. Thus, we observe that the NE solution is near-optimal in small networks, which is mainly due to that the D2D links are active with probability close to one in most cases in the small network. On the other hand,  the D2D links have fractional active probabilities in most cases in the large network, and thus the observation may be quite different in large networks. We leave the analysis of the gap between the NE and optimal solution of (\ref{eq:D2Dopt}) in more general networks to future work. To compare the BR Algorithm and LB Algorithm, we consider a case with ten D2D links. Fig.~\ref{fig:ratecdf-org} shows that the rate distributions using the BR Algorithm and the LB Algorithm are almost the same. This implies that we can use the solution of the LB Algorithm to approximate the solution of the BR Algorithm. By comparing to  Figs. \ref{fig:RCcdf} and \ref{fig:RDcdf}, we  observe that the performance of different algorithms in single-cell networks is similar to the multi-cell networks. Therefore, more discussion is left to the following subsection. 
\begin{figure}
\centering
\includegraphics[width=8cm,height=6cm]{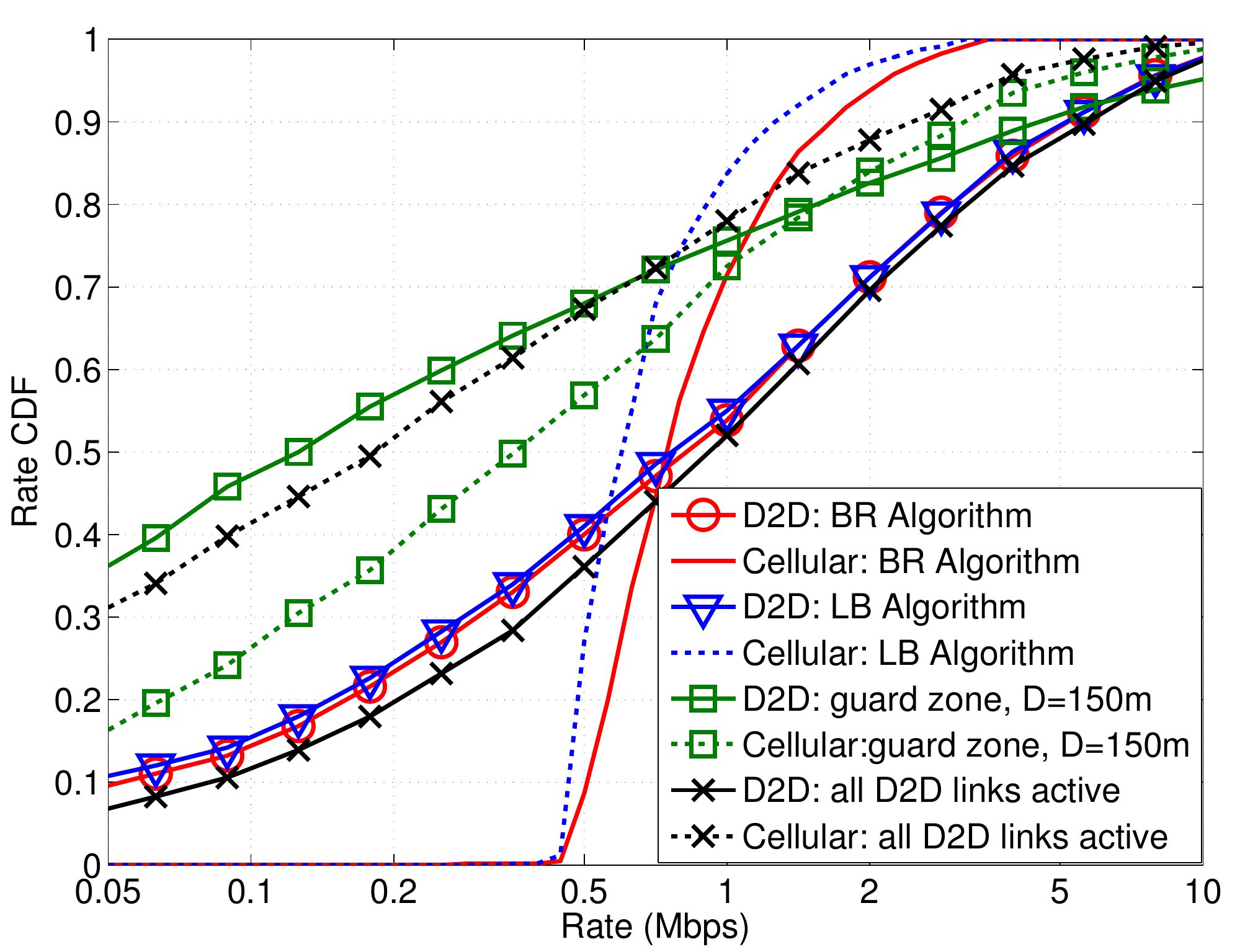}
\caption{The rate distributions of D2D and cellular links using different algorithms in a single-cell network.} 
\label{fig:ratecdf-org}
\end{figure}

\subsection{The Upper Problem: Network Pricing Mechanism}
\begin{figure}
\centering
\includegraphics[width=8cm,height=6cm]{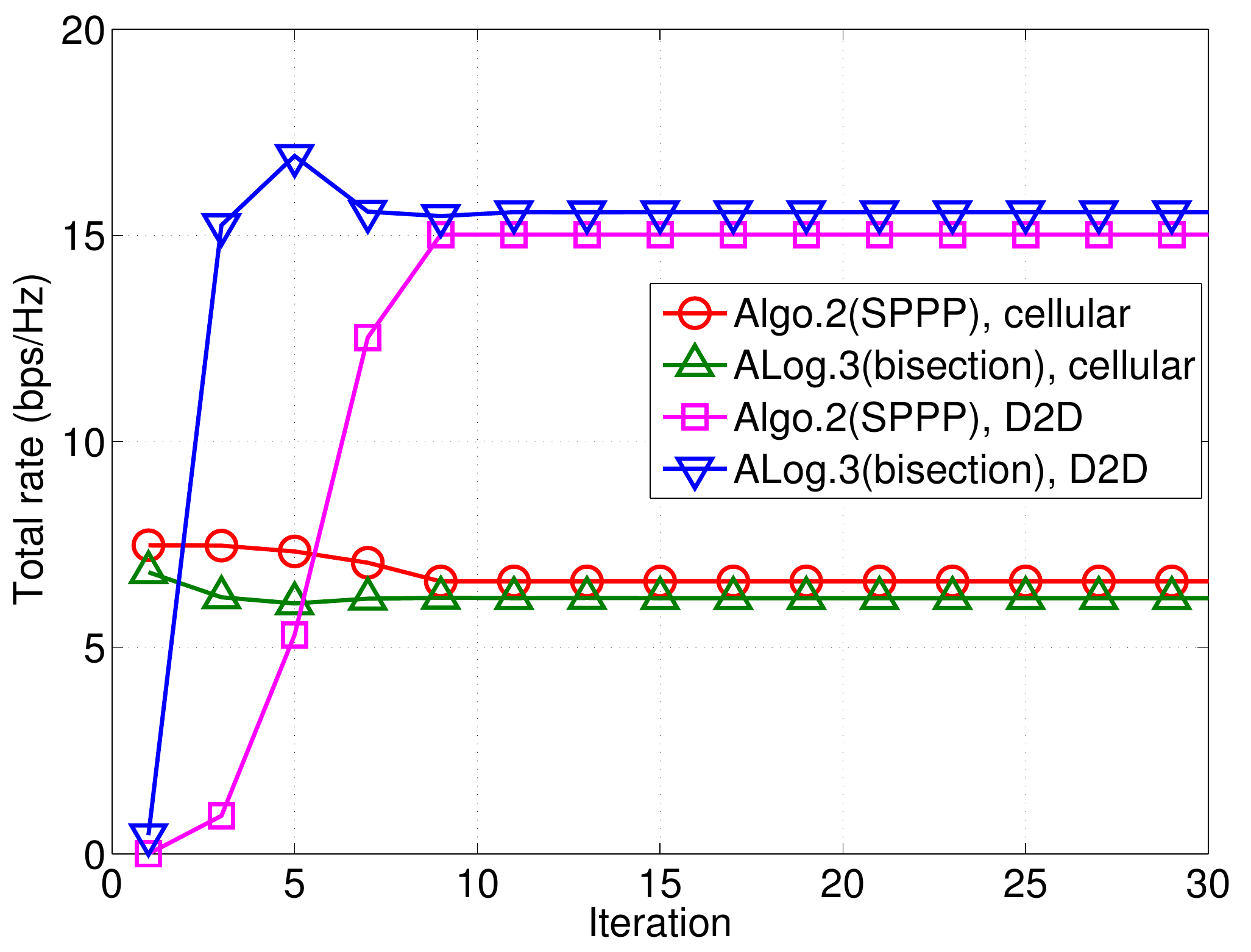}
\caption{The convergence of  different algorithms.}
\label{fig:converge-game}
\end{figure}
In this section, we consider an asynchronous  multi-cell network, where each cell allocates resources independently. We let the interference tolerance level be the same as the received signal of cellular link (i.e., the interference level normalized by the cellular signal is 0dB). From the simulation results, the number of iterations required for the convergence of the LB Algorithm is about 4-8, which is quite small. Fig. \ref{fig:converge-game} shows the convergence of the SPPP and bisection algorithms. Both algorithms converge quickly.  While SPPP provides a larger cellular rate, it converges more slowly than the bisection~algorithm. The quick convergence of LB Algorithm and bisection algorithm implies that the complexity of the proposed scheme  $O(N_D\log_2(\mu_{\max} / \epsilon) + N_D^2 K)$ is much lower than the complexity of the centralized scheme  $O(N_x^{N_D} N_D^2)$, where $K\in [4,8]$ and $\log_2(\mu_{\max} / \epsilon)\in [5,10]$, while $N_x\geq 2$ generally and thus $N_x^{N_D}\geq 2^{10}$ in our setup.

\begin{figure}
\centering
\includegraphics[width=8cm, height=6cm]{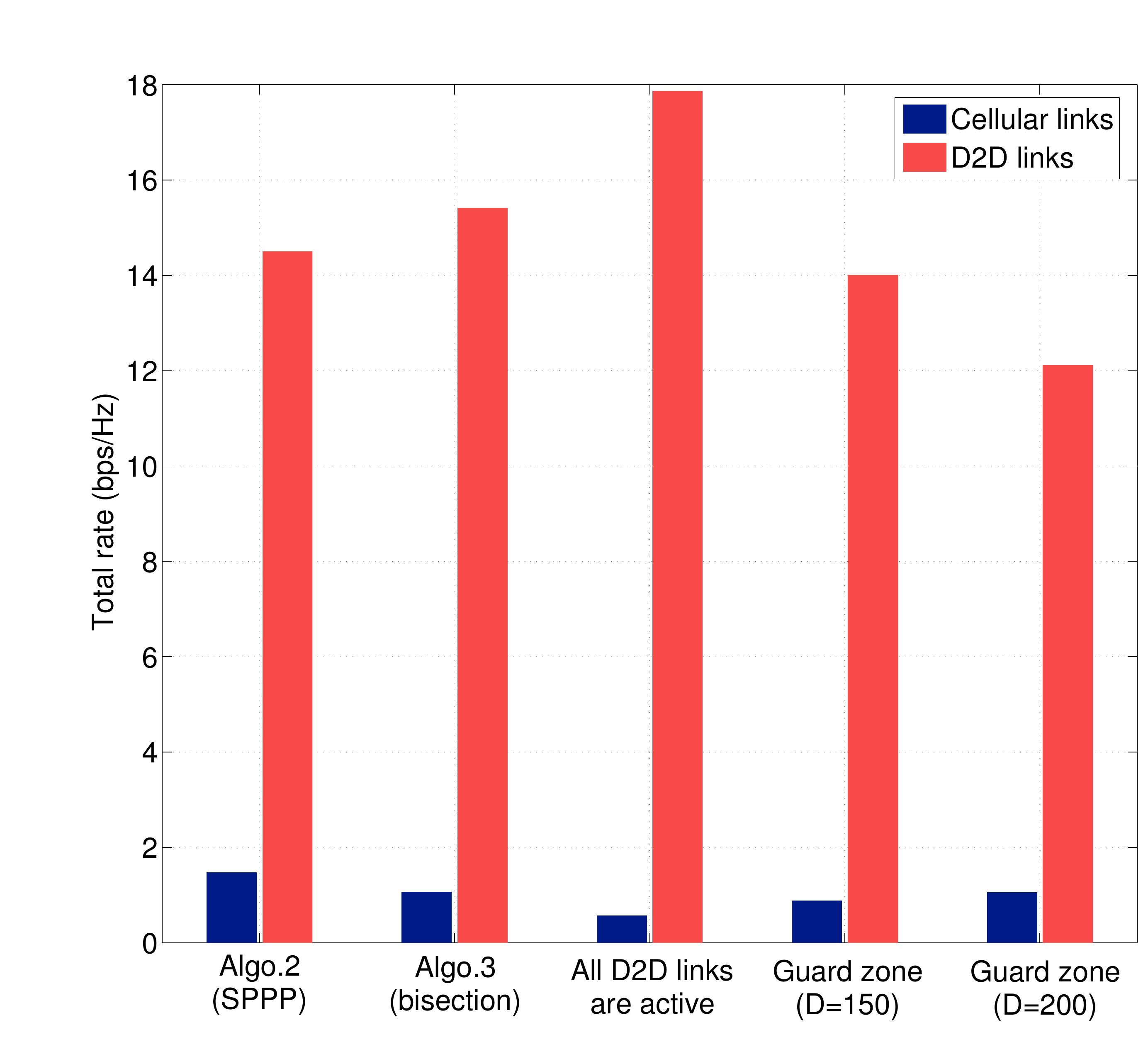}
\caption{The total rates of cellular and D2D links using different approaches.}
\label{fig:sumrate-shared}
\end{figure}

In Fig. \ref{fig:sumrate-shared}, we compare the rates of D2D and cellular links using different algorithms. The  SPPP and bisection algorithms provide larger D2D and/or cellular rates than the guard zone schemes. If there is no interference management, the rate of cellular UE is very small (see ``All D2D active'' in the figure). Adopting the proposed algorithms, the cellular links can get much better performance (total cellular rate increasing from $0.61$ to about $1.07$ bps/Hz), at the cost of less total throughput (about $12\%$ loss in our setup). We observe that  SPPP provides a slightly larger rate for cellular links than the bisection algorithm. This implies that in some cases, the function $U_{c1}$ is non-monotonic and the optimal $\mu$ is not at the intersection of $U_{c1}$ and $U_{c2}$. However, in general, the gap between the bisection algorithm and the SPPP algorithm is small regardless of the monotonicity of  $U_{c1}$. The average total rate in conventional networks, where potential D2D links operate only in cellular mode, is $2.4$ bps/Hz in our setup. Defining the rate gain by the increased total rate divided by the average rate in conventional networks, we conclude that allowing D2D links and using proposed algorithms achieves a very large rate gain compared to conventional networks (about $5$x in our simulation setup), and meanwhile keeps the performance of cellular UEs at an acceptable level (with average rate per cellular link being $1.07$ bps/Hz). Note that the rate gain depends on various system parameters, such as average D2D link length and the amount of D2D traffic.

The rate distributions of cellular and D2D links are shown in Figs. \ref{fig:RCcdf} and \ref{fig:RDcdf}, respectively. Fig.~\ref{fig:RCcdf} shows that proposed algorithms can effectively protect the cellular performance. Comparing to Fig.~\ref{fig:ratecdf-org}, where the average cellular rate is about $1.5$ bps/Hz with the normalized interference tolerance level being $5$dB, we can conclude that a lower normalized interference tolerance level ($0$dB) is needed in the multi-cell scenario. Also, the rate of cellular links has a larger range than the single cell scenario (i.e., the variance is larger). One possible reason is that there may be some nearby interfering D2D links and cellular UEs in neighboring cells. Though the D2D links have large rates without interference management, they hurt cellular links a lot. Adopting the guard zone scheme, cellular links can be protected, at the cost of the degradation of D2D throughput. Moreover, it is difficult to develop a tractable framework to study the guard zone scheme, and thus difficult to find the optimal distance threshold analytically. Therefore, the SPPP and bisection schemes are more preferable.

\begin{figure}
\centering
\includegraphics[width=8cm, height=6cm]{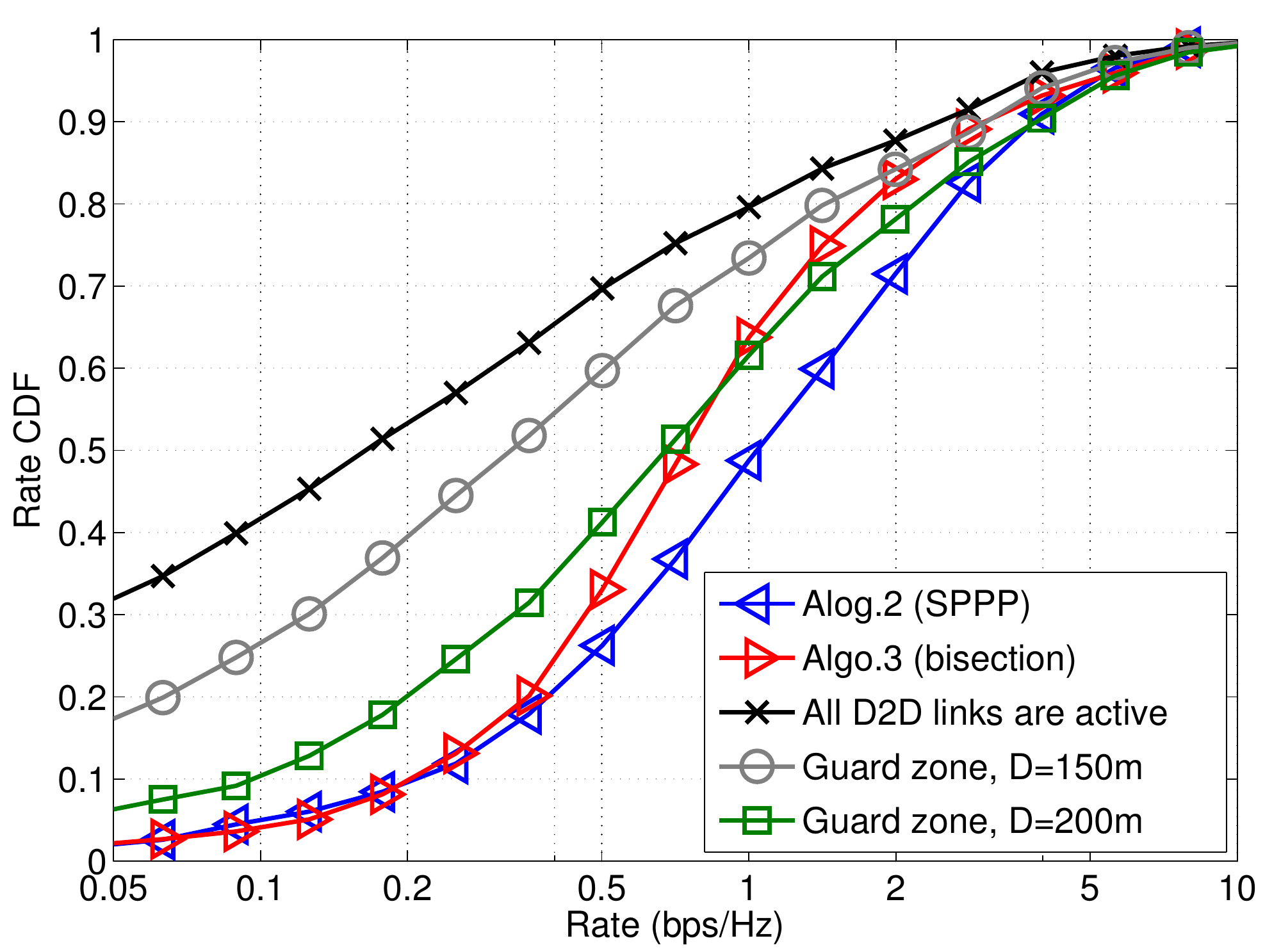}
\caption{The rate distribution of cellular links using different approaches. }
\label{fig:RCcdf}
\end{figure}

\begin{figure}
\centering
\includegraphics[width=8cm, height=6cm]{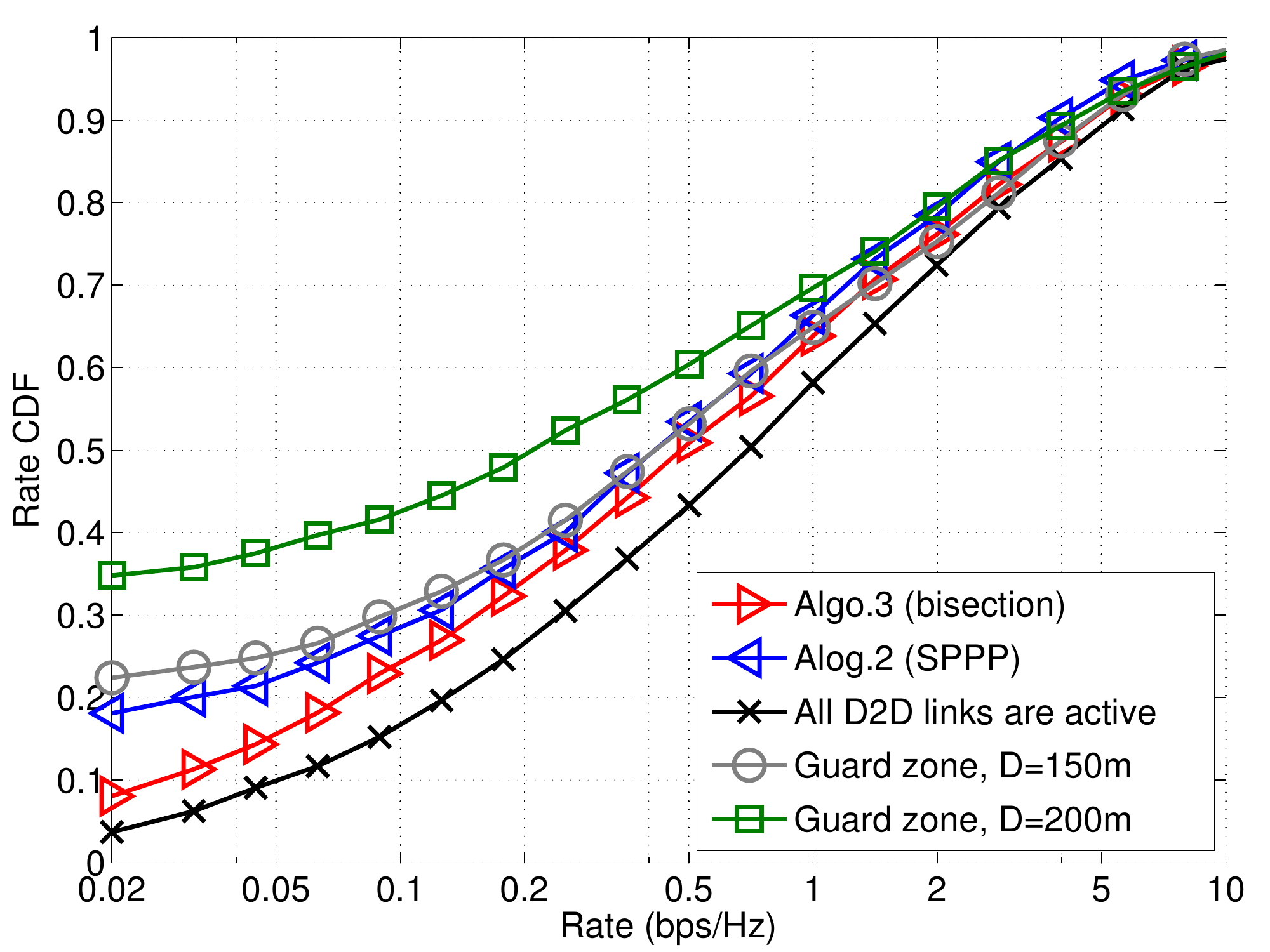}
\caption{The rate distribution of D2D links using different approaches. }
\label{fig:RDcdf}
\end{figure}

Note that the interference tolerance level can be tunable to maximize utility functions in terms of both the cellular and D2D links (e.g., the total rate in the hybrid network). 
We show the rates of cellular and D2D links versus the interference tolerance level numerically in Figs.~\ref{fig:RCvsQ} and \ref{fig:RDvsQ}, respectively. The analysis of optimal $Q$ with respect to different utility functions is left to future work. Fig. \ref{fig:RCvsQ}  shows that as the interference tolerance level increases, the rate of cellular users decreases, because more D2D links are allowed to transmit. On the other hand, as $Q$ increases, D2D links can access the RBs more aggressively and the total rate of D2D links increases. 
The IO Algorithm protects the performance of cellular links well. However, in a network with strict interference constraints, the total rate of D2D links using the IO Algorithm is less than the SPPP and bisection algorithms, which implies the importance to consider power control for D2D resource allocation, as well as the interference experienced at D2D receivers. From Fig.~\ref{fig:RDvsQ}, we can conclude that although the IO Algorithm is very simple, it can only be applied to the cases with high interference tolerance (e.g., cases with normalized interference tolerance level larger than $0$dB).

\begin{figure}
\centering
\includegraphics[width=8cm, height=6cm]{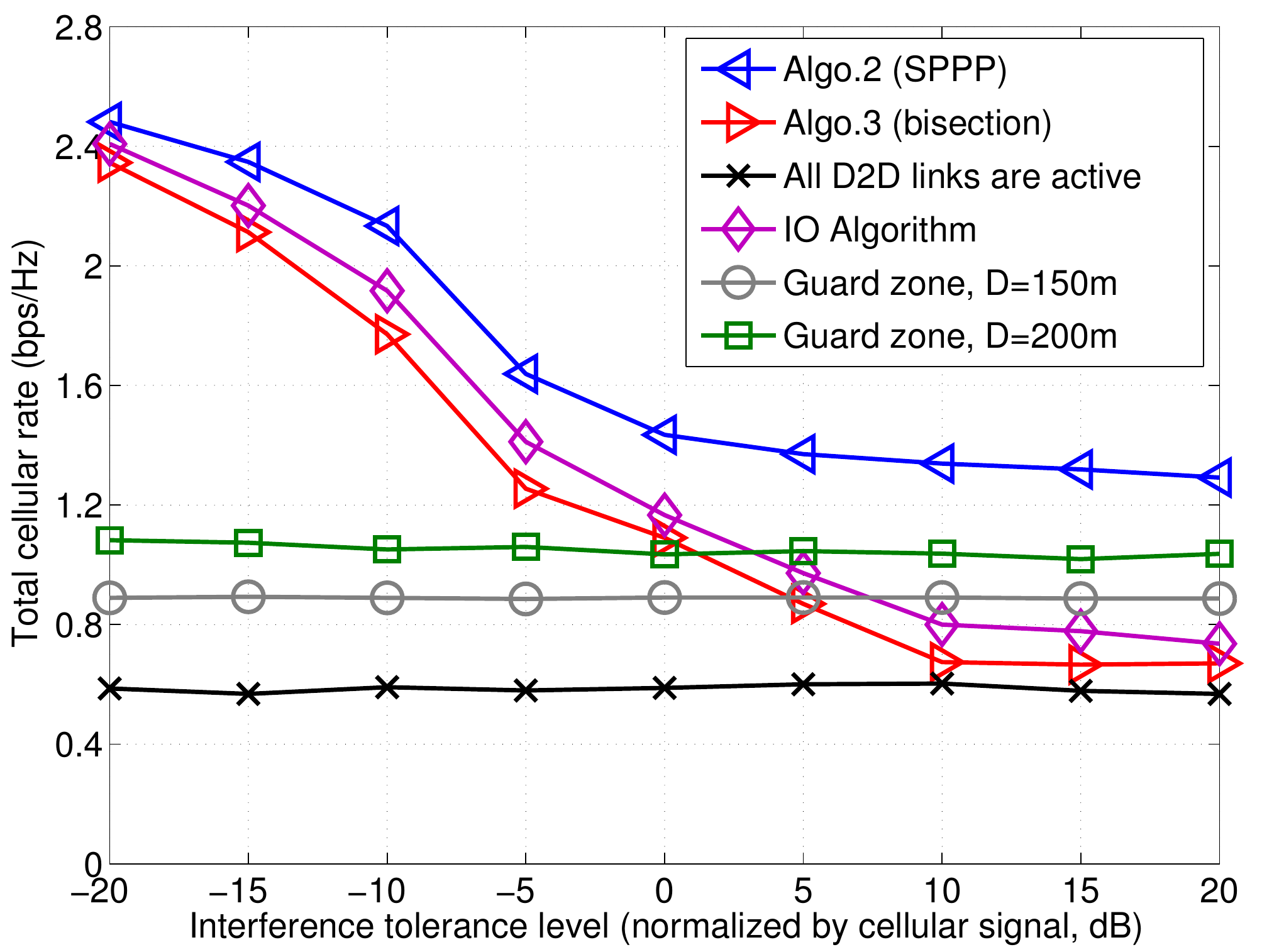}
\caption{The  rate of cellular links vs. the interference tolerance level. The normalized interference tolerance level means that the interference tolerance level $Q$ is divided by the signal of the cellular link accessing the considered~RB.}
\label{fig:RCvsQ}
\end{figure}

\begin{figure}
\centering
\includegraphics[width=8cm, height=6cm]{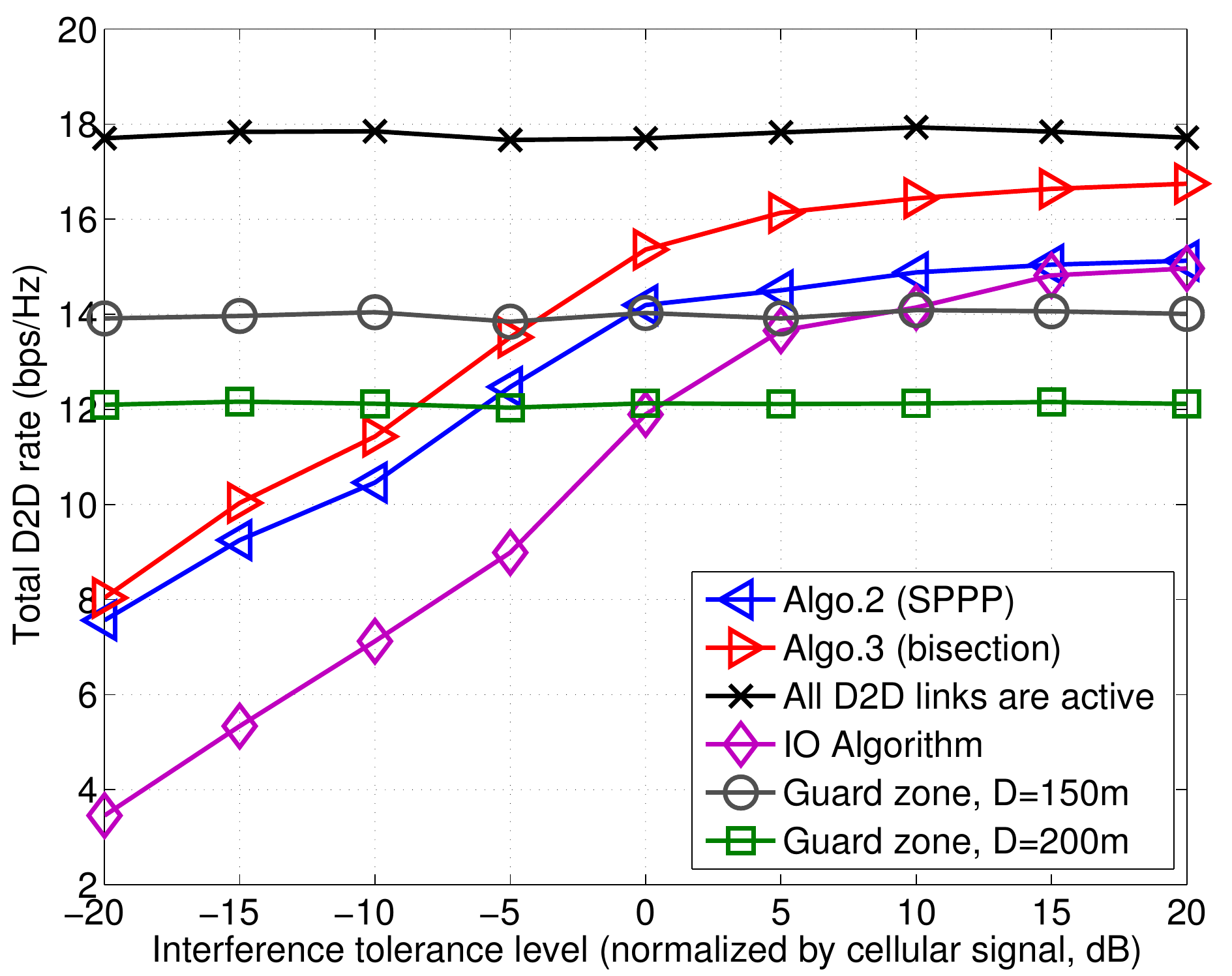}
\caption{The total rate of D2D links vs. the interference tolerance level.  The normalized interference tolerance level means that the interference tolerance level $Q$ is divided by the signal of the cellular link accessing the considered~RB.}
\label{fig:RDvsQ}
\end{figure}

Figs. \ref{fig:RCvsDense} and \ref{fig:RDvsDense} show the rates of cellular links and D2D links versus the D2D density, respectively. We ignore the guard zone scheme with radius 150m in these figures, since its performance is similar to the guard zone scheme with radius 200m. As shown in Fig. \ref{fig:RDvsDense}, the total rate of D2D links increases as D2D density increases, while the rate of cellular link decreases in Fig. \ref{fig:RCvsDense}. The decrease of cellular rate using the SPPP and bisection algorithms vanishes much more quickly than the scenario with all D2D links being active, which suggests the efficiency of the SPPP and bisection algorithms for protecting cellular transmissions. 
The figures also show that besides the interference tolerance level, the throughput gain of the proposed algorithms also highly depends on the density of  D2D links.

\begin{figure}
\centering
\includegraphics[width=8cm, height=6cm]{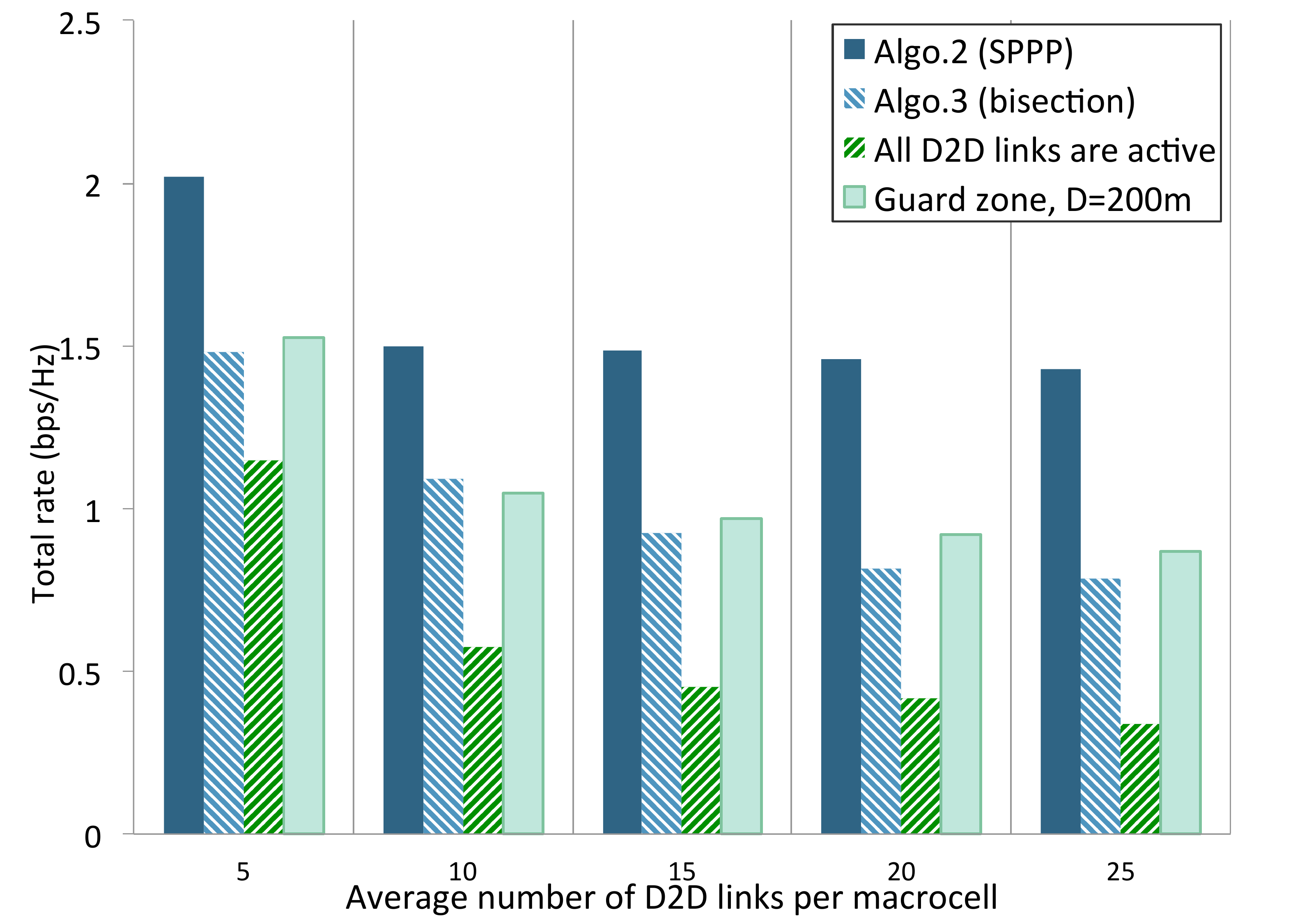}
\caption{The total rate of cellular links vs. different D2D densities.}
\label{fig:RCvsDense}
\end{figure}

\begin{figure}
\centering
\includegraphics[width=8cm, height=6cm]{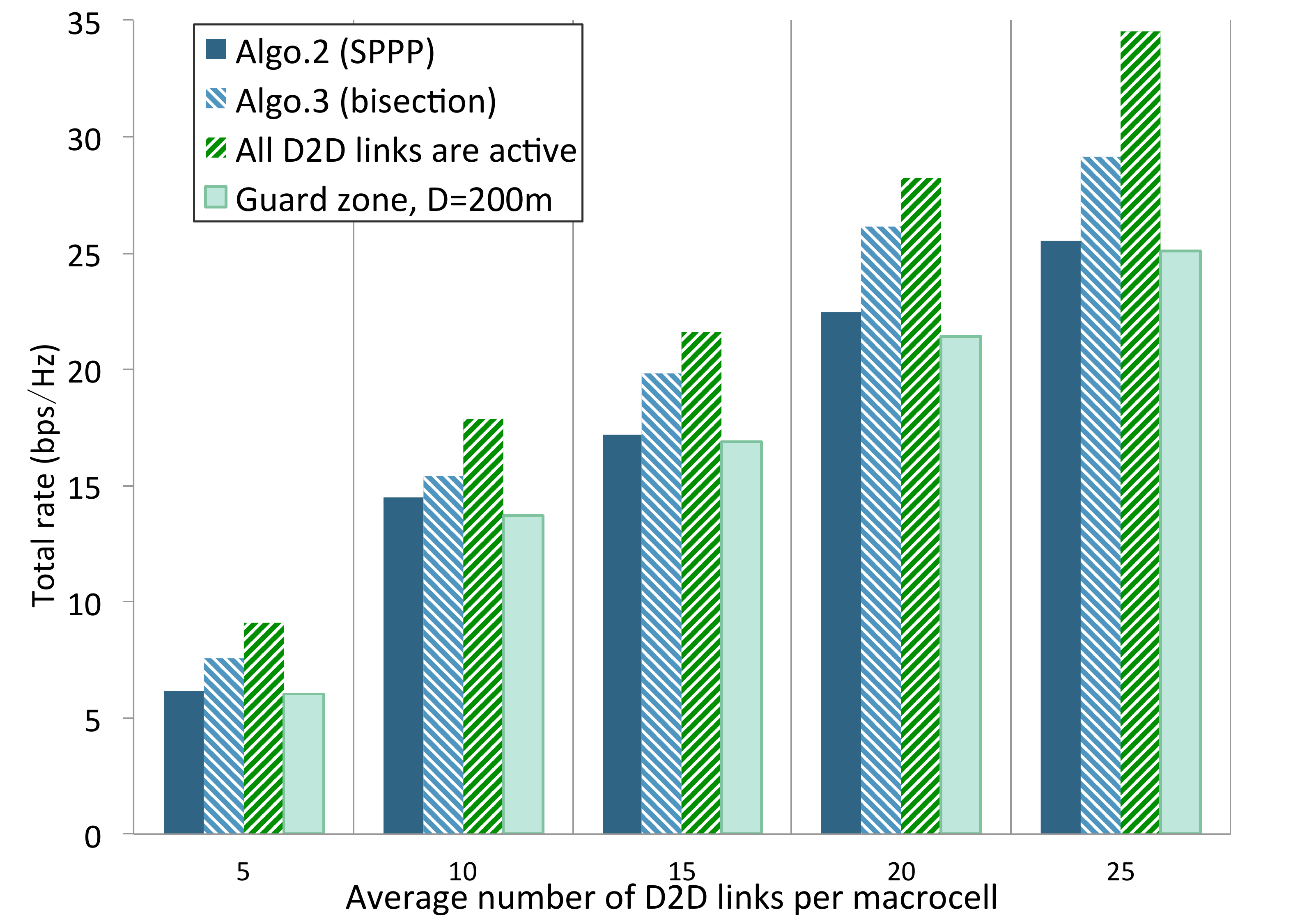}
\caption{The total rate of D2D links vs. different D2D densities.}
\label{fig:RDvsDense}
\end{figure}

\section{Conclusion}\label{sec:conclusion}
This paper presents a decentralized spectrum management for a shared network consisting of D2D and cellular links, aiming to maximize the total throughput of D2D links with an interference constraint for protecting cellular transmissions. We propose a low-complexity low-overhead distributed algorithm to update D2D access probabilities, and use the SPPP algorithm to get the optimal price for controlling the interference from co-channel D2D links. Though the SPPP Algorithm requires global CSI, it provides a benchmark for other algorithms. We further propose a low-overhead efficient heuristic algorithm based on the bisection method, which is shown to be convergent. Numerical results show that the heuristic algorithm has about the same performance as the SPPP algorithm, especially in the cases with low interference tolerance level. Another simple greedy algorithm is proposed and shown to perform well in scenario with high interference tolerance level. The proposed algorithms provide a large throughput gain with a performance guarantee of cellular links, compared to a conventional network with links operating only in cellular mode. Comparing to the cases without interference management (i.e., all D2D links are active), the average rate of cellular links improves significantly (e.g., average rate per cellular link increases from $0.61$ to $1.07$ bps/Hz in our setup). This implies that the proposed algorithms can efficiently  manage the interference from D2D links to the cellular network. Future work could include investigation of  more general utility functions incorporating both throughput and fairness, joint optimization of D2D mode selection, and consideration of a more flexible multiple-cell system.

\appendices
\section{Proof of Proposition \ref{prop:BRfunc-property}}\label{pf:prop-BRfunc-property}
\begin{enumerate}
\item  We can complete the proof by applying the \emph{Maximum Theorem} with $\Phi=U_i(\mathbf{x}_i;\mathbf{x}_{-i},\boldsymbol{\mu})$ \cite{ShuLeu07}.
\item 
Let $\mathcal{A}$ and $\mathcal{S}$ be any two distinct subsets of set $\mathcal{D}=\{1, 2, \cdots, N_D\}$. We have
\begin{equation}\label{eq:def-a-s}
\left\{
\begin{aligned}
& a_{i}-s_{i}\leq 0, && \textrm{for } i\in\mathcal{D}\setminus\left(\mathcal{A}\cup \mathcal{S} \right),\\
& 0<a_{i}-s_{i}<P_{D_i}h_{ii}, && \textrm{for } i\in\mathcal{A},\\
&a_{i}-s_{i} \geq P_{D_i}h_{ii}, && \textrm{for } i\in\mathcal{S}.
\end{aligned}
\right.
\end{equation}

Given that the number of D2D links is finite, we have that the number of choices of disjoint $\mathcal{A}$ and $\mathcal{S}$ is finite. Therefore, the domain of $f_L$ can be partitioned into finitely many polyhedra. With $\mathbf{s}=\mathbf{Gx}$, the inequalities (\ref{eq:def-a-s}) can be changed to inequalities in terms of $\mathbf{x}$, which defines a (possibly empty) polyhedron in $\mathcal{X}$. On the polyhedron $\mathcal{P}_n$, according to (\ref{eq:optsol-D2D}), we have $x_{i} =\frac{a_{i}-s_{i}}{P_{D_i}h_{ii}}$ for $i\in\mathcal{A}^{(n)}$, $x_i =1$ for $i\in\mathcal{S}^{(n)}$ and $x_i=0$ otherwise, which can be expressed in a matrix form as $x_{i}=\mathbf{B}_i^{(n)} s_i + \mathbf{b}_i^{(n)}$, where $\mathbf{B}_{i}^{(n)}$ is defined in the Proposition \ref{prop:BRfunc-property}. Combining with $\mathbf{s}=\mathbf{Gx}$, we complete the proof.

\end{enumerate}

%
%

\section{Proof of Proposition \ref{prop:U1-property}}\label{pf:prop-U1-property}
\begin{enumerate}
\item We define a function $g:\mathbb{R}^+ \rightarrow \mathcal{X}$ as $g(\mu) = \left(g_1(\mu), \dots, g_{N_D}(\mu)\right)$, where $g_i(\mu)=x_i(\mu,\mathbf{x}_{-i}(0))$,  $\mathbf{x}(0)$ is a given initial vector, and $x_i(\mu, \mathbf{x}_{-i}(0))$ is calculated according to (\ref{eq:optsol-D2D}) by fixing $\mathbf{x}_{-i} = \mathbf{x}_{-i}(0)$. Observing (\ref{eq:optsol-D2D}), we can see that $g_i(\mu)$ is a continuous function for a given $\mathbf{x}_{-i}$. According to properties of continuous functions (see, e.g., Theorem 4.10 in \cite{Rudin64}), $g(\mu)$ is continuous due to the fact that each of the function $g_1(\mu), \dots, g_{N_D}(\mu)$ is continuous.  Proposition \ref{prop:BRfunc-property} shows that the best-response function $f_L:\mathcal{X}\rightarrow \mathcal{X}$ is continuous. Invoking properties of continuous functions again (see, e.g., Theorem 4.7 in \cite{Rudin64}), we can conclude that $f_L(g(\mu)$ is a continuous mapping, which implies that the NE $x_i^*(\mu)$ is a continuous function of $\mu$. Therefore, $U_{c1}$ is also a continuous function of $\mu$.

\item Recall that $\mathcal{A}$ and $\mathcal{S}$ denote the sets of active D2D links and of saturated D2D links, respectively.  Without loss of generality, let $\mathcal{A}=\{1,\dots,n\}$ and $\mathcal{S}=\{n+1, \dots, n+m\}$.  We denote $\mathbf{x}_{a}=[x_1,  \cdots, x_n]^T$, $\mathbf{x}_{s}=[x_{n+1},  \cdots, x_{n+m}]^T$ and  $\mathbf{x}_{0}=[x_{n+m+1}, \cdots, x_{N_D}]^T$. Let $\mathbf{H}_{aa}$  be a matrix with $(i,j)$th element being $\frac{P_{D_j}h_{ji}}{P_{D_i}h_{ii}}$ for $i,j \in\mathcal{A}$, $\mathbf{W}_{a} = \left[\frac{w_1}{P_{D_1}g_{11}},   \cdots,   \frac{w_n}{P_{D_n}g_{nn}}\right]^T$ and $\mathbf{C}_{a} = \left[\frac{I_{C1} + I_{D_1}}{P_{D_1}h_{11}},   \cdots,  \frac{I_{Cn} + I_{D_n}}{P_{D_n}h_{nn}}\right]^T$, where $I_{D_i} = \sum_{j\in\mathcal{S}} P_{D_j}h_{ji}$. According to~(\ref{eq:optsol-D2D}), we have $\mathbf{x}_{a} = (\mathbf{H}_{aa})^{-1}\left[\frac{\mathbf{W}_a}{\mu} - \mathbf{C}_a \right]$,  $\mathbf{x}_s = 1$ and $\mathbf{x}_0 = 0$.  The domain of function $U_{c1}$ can be divided into finite polyhedra according to different $\mathcal{A}$ and $\mathcal{S}$.  In each polyhedron, we have
\begin{equation}\label{eq:U1-matrix}
\begin{aligned}
&U_{c1}(\mu)
=\mu \sum_{i\in\mathcal{A}} x_{i}P_{D_i}g_{ii} + \mu \sum_{i\in\mathcal{S}} P_{D_i}g_{ii}\\
&=\mu \boldsymbol{\beta}^T_a (\mathbf{H}_{aa})^{-1} \left[ \frac{\mathbf{W}_a}{\mu}-\mathbf{C}_a \right] + \mu \boldsymbol{\beta}^T_s \mathbf{1}^T\\
& =  \boldsymbol{\beta}^T_a (\mathbf{H}_{aa})^{-1}  \left[ \mathbf{W}_a +\left(\mathbf{H}_{aa}\boldsymbol{\tilde{\beta}}_1 \left(\boldsymbol{\beta}^T_s \mathbf{1}^T \right) -\mathbf{C}_a\right)\mu \right],
\end{aligned}
\end{equation}
where $\boldsymbol{\beta}_a=\left[P_{D_1}g_{11},P_{D_2}g_{22} \cdots, P_{D_n}g_{nn} \right]^T$, $\boldsymbol{\beta}_s=\left[ P_{D_{n+1}}g_{(n+1)(n+1)}, \cdots, P_{D_{n+m}}g_{(n+m)(n+m)} \right]^T$, $\boldsymbol{\tilde{\beta}}_1=[\frac{1}{P_{D_1}g_{11}}, 0, \cdots, 0]$, and $\mathbf{1}=[1, 1, \cdots, 1]^T$. Therefore, $U_{c1}$ is a linear function in each given polyhedron, and thus it is piecewise affine.

\item We use the same argument as the proof of Theorem 1 in \cite{RazLuo11}. The first condition  is equivalent to  that  matrix $\mathbf{H}_{aa}^{(\boldsymbol{\beta})} \defeq \textrm{diag} (\boldsymbol{\beta}_a) \mathbf{H}_{aa}\left(\textrm{diag}(\boldsymbol{\beta}_a) \right)^{-1}$ is strictly (column-wise) diagonally dominant. Defining $\mathbf{C}_{a}^{(\boldsymbol{\beta})} =\textrm{diag}(\boldsymbol{\beta}_a) \left(\mathbf{C}_a - \mathbf{H}_{aa}\boldsymbol{\tilde{\beta}}_1 \left(\boldsymbol{\beta}_s^T \mathbf{1}^T\right)\right)$, we have
\begin{equation} \label{eq:cof-mu}
 \boldsymbol{\beta}_a^T \mathbf{H}_{aa}^{-1} \left(\mathbf{C}_a - \mathbf{H}_{aa}\boldsymbol{\tilde{\beta}}_1 \left(\boldsymbol{\beta}_s^T \mathbf{1}^T\right)\right) = \mathbf{1}^T\left(\mathbf{H}_{aa}^{(\boldsymbol{\beta})} \right)^{-1} \mathbf{C}_{a}^{(\boldsymbol{\beta})}.
\end{equation}
To show that $U_{c1}$ is non-increasing, we need to show that (\ref{eq:cof-mu}) is non-negative. The first condition is a sufficient condition for $ \mathbf{1}^T\left(\mathbf{H}_{aa}^{(\boldsymbol{\beta})} \right)^{-1} \geq 0$. Similar proof can be found in \cite{RazLuo11}, and we ignore the details. 
The remaining proof is to show  $\mathbf{C}_a - \mathbf{H}_{aa}\boldsymbol{\tilde{\beta}}_1 \left(\boldsymbol{\beta}_s^T \mathbf{1}^T\right)\geq 0$.  The $i$th element of the left term of the above inequality is $\frac{1}{P_{D_i}h_{ii}} \left(I_{C_i} + \sum_{j\in\mathcal{D}_s}P_{D_j}h_{ji}- \frac{h_{1i}}{g_{11}}  \sum_{j\in\mathcal{D}_s}P_{D_j}g_{jj} \right)$, which
 implies that $\sum_{j\in\mathcal{D}_s}P_{D_j}\left(h_{ji}- \frac{h_{1i}}{g_{11}} g_{jj}\right)\geq 0, \ \forall i\in\mathcal{D}_a$ is a sufficient condition to make $\mathbf{C}_a - \mathbf{H}_{aa}\boldsymbol{\tilde{\beta}}_1 \left(\boldsymbol{\beta}_s^T \mathbf{1}^T\right)\geq 0$. Combining with $\mathbf{1}^T\left(\mathbf{H}_{aa}^{(\boldsymbol{\beta})} \right)^{-1} \geq 0$, 
we can conclude that (\ref{eq:cof-mu}) is non-negative, and thus $U_{c1}$ is non-increasing.

\end{enumerate}

\section{Proof of Proposition \ref{prop:bisection}}\label{pf:prop-bisection}
We first show that the intersection of $U_{c1}$ and $U_{c2}$ on  $[0,\mu_{\max}]$ is non-empty, and then show that the bisection algorithm converges to one of the intersection points. Prop. \ref{prop:U1-property} shows that $U_{c1}$ is a continuous function of $\mu$. It is easy to observe that $U_{c2}$ is also a continuous function of $\mu$. Therefore, $U_{c1}-U_{c2}$ is a continuous function of $\mu$. Recalling the assumption that when all D2D links are active, the interference from D2D to BSs is greater than the interference tolerance level, we have $U_{c1}-U_{c2}>0$ when $\mu=0$. On the other hand, when $\mu=\mu_{\max}$, we have $U_{c1} -U_{c2}<0$. According to the intermediate value theorem, we can conclude that there is some number $\mu\in[0,\mu_{\max}]$ such that $U_{c1}(\mu) - U_{c2} (\mu)=0$. In other words, there is at least one intersection point between $U_{c1}$ and $U_{c2}$ on $[0,\mu_{\max}]$. 

Adopting the bisection algorithm, the interval is divided into two halves at each iteration. The interval at iteration $t$ is denoted by $L(t)=[a_t, b_t]$, where $a_0=0$ and $b_0=\mu_{\max}$. According to procedures of the bisection algorithm, we have $U_{c1}(a_t)\geq U_{c2}(a_t)$ and $U_{c1}(b_t)\leq U_{c2}(b_t)$ at each iteration $t$. Similar to the proof of the existence of intersection points on $[0, \mu_{\max}]$, we can show that there is at least one intersection point between $U_{c1}$ and $U_{c2}$ on $L(t)$. Therefore, the bisection algorithm preserves the existence of intersection points in current interval. The length of interval $L(t)$ has $|L(t)|=|L(t-1)|/2=\dots = \mu_{\max}/2^t$. It must stop when $|L(t)|\leq \epsilon$, which implies that the algorithm converges, and the maximum number of iteration for convergence, denoted by $T$, satisfies $\mu_{\max}/2^T = \epsilon$, i.e.,  $T=\log_2(\mu_{\max}/\epsilon)$.


\bibliographystyle{ieeetr}
\bibliography{d2dgamebib}

\end{document}